\documentclass[a4paper,12pt]{article}
\usepackage[bottom=4.5cm,top=3.8cm,left=2.3cm,right=2.3cm]{geometry}
\usepackage{graphicx}
\usepackage{subcaption}
\usepackage[shortlabels]{enumitem}
\usepackage{url}
\usepackage{hyperref}
\usepackage{multirow}
\usepackage{tabularx}
\usepackage{booktabs}

\usepackage{amsthm}
\usepackage{amsmath}
\usepackage{amssymb}
\usepackage{chngcntr}

\usepackage{algorithmicx}
\usepackage{algpseudocode}
\usepackage{algorithm}

\numberwithin{equation}{section}

\newcolumntype{Y}{>{\centering\arraybackslash}X}
\newcommand\Tstrut{\rule{0pt}{2.9ex}}       
\newcommand\Bstrut{\rule[-0.5ex]{0pt}{0pt}} 
\newcommand{\TBstrut}{\Tstrut\Bstrut} 

\newcommand{\floor}[1]{\lfloor{#1}\rfloor}
\newcommand{\ceil}[1]{\lceil{#1}\rceil}
\newcommand{\dist}{{\mathop{\rm dist}}}
\newcommand{\FP}{{\mathop{\rm FP}}}
\newcommand{\SFP}{{\mathop{\rm SFP}}}
\newcommand{\SP}{{\mathop{\rm SP}}}

\newtheorem{theorem}{Theorem}[section]
\newtheorem{lemma}{Lemma}[section]
\newtheorem{property}{Property}[section]

\newtheorem{corollary}{Corollary}[section]

\theoremstyle{definition}

\newtheorem{assumption}{Assumption}[section]

\begin{document}

\begin{center}
{\Large Algorithms for the Ridesharing with Profit Constraint Problem}
\vskip 0.2in
Qian-Ping Gu$^1$, Jiajian Leo Liang$^1$

$^1$School of Computing Science, Simon Fraser University, Canada\\
qgu@sfu.ca, leo\_liang@sfu.ca
\end{center}
\vspace{2mm}

\noindent \textbf{Abstract:}
Mobility-on-demand (MoD) ridesharing is a promising way to improve the occupancy rate of personal vehicles and reduce traffic congestion and emissions. Maximizing the number of passengers served and maximizing a profit target are major optimization goals in MoD ridesharing.
We study the ridesharing with profit constraint problem (labeled as RPC) which considers both optimization goals altogether: maximize the total number of passengers subject to an overall drivers' profit target.
We give a mathematical formulation for the RPC problem.
We present a polynomial-time exact algorithm framework (including two practical implementations of the algorithm) and a $\frac{1}{2}$-approximation algorithm for the case that each vehicle serves at most one passenger.
We propose a $\frac{2}{3\lambda}$-approximation algorithm for the case that each vehicle serves at most $\lambda \geq 2$ passengers.
Our algorithms revolve around the idea of maximum cardinality matching in bipartite graphs and hypergraphs (set packing) with general edge weight. Based on a real-world ridesharing dataset in Chicago City and price schemes of Uber, we conduct an extensive empirical study on our model and algorithms. Experimental results show that practical price schemes can be incorporated into our model, our exact algorithms are efficient, and our approximation algorithms achieve $\sim$90\% of optimal solutions in the number of passengers served.
\vspace{2mm}


\section{Introduction} \label{introduction}
Personal vehicles and mobility-on-demand (MoD) systems are major transportation tools worldwide. MoD systems, such as Uber, Lyft and DiDi, have become popular around the globe due to their convenience.
Drivers and passengers can be matched based on ridesharing requests (arrange a ride for passengers in a personal vehicle).
MoD system operators and drivers participated in such systems are mostly motivated by profit in practice. Solely focusing on profit and market share from MoD systems and drivers may have increased congestion and $\textup{CO}_2$ emissions; the use of MoD has increased the number of single-passenger vehicles on the road significantly~\cite{Diao-NS21,Henao-Trans19,Tirachini-IJST20}.
This, coupled with the saturated personal vehicle usage (with low occupancy rate) in Europe 
and North America, causes more traffic congestion and emissions~\cite{Martins-CIE21}. According to studies in~\cite{CSS20, Sierpinski-ATST13, StatsCan-2016}, personal vehicles were the main transportation mode in the United States and Canada in recent years and in more than 200 European cities between 2001 and 2011. In Europe 2017~\cite{EEA19}, the transport sector accounted for 27\% of total greenhouse gas emissions; and of these 27\% gas emissions, 31.55\% (8.52\% total) were from passenger cars.
The occupancy rate of personal vehicles in the U.S. was 1.6 persons per vehicle in 2011~\cite{Ghoseiri-USDT11,Santos-USDTFHA11} (and decreased to 1.5 persons per vehicle in 2017~\cite{CSS20}).

On the other hand, there is an urgency to reduce traffic congestion and greenhouse gas emissions. Ridesharing using MoD systems has been proposed and studied in the academia~\cite{Agatz-EJOR12,Furuhata-TRBM13,Martins-CIE21,Mourad-TRBM19,Tafreshian-SS20,Wang-TRBM19}. Major themes from many previous studies include maximizing the total number of passengers served, minimizing the total cost to serve all passengers and maximizing a profit target.
Studies, such as those in \cite{Alonso-Mora-PNAS17, Amirkiaee-TRFTPB18,Furuhata-TRBM13,Tikoudis-TRDTE21}, have shown that ridesharing is a promising effective way to increase the occupancy rate and reduce congestion. It is estimated that commuting to work by ridesharing in Dublin, Ireland, can reduce 12,674 tons of CO$_2$ emissions per year~\cite{Caulfield-TRDTE09}, and taxi-ridesharing in Beijing can reduce 120 million liters of gasoline annually~\cite{Ma-TKDE15}.
An important factor for the adoption of ridesharing in practice is the profit/pricing scheme.
Demand-and-pricing of a ridesharing system is important for its adaptability of actual ridesharing in practice~\cite{Wang-TRBM19}.
Recently, ridesharing with profit as taxi ridesharing (e.g., \cite{Jung-CACIE16,Ma-TKDE15,Qian-TRBM17,Santos-ESA15} and pricing based platform (MoD) equilibrium analysis (e.g., \cite{Besbes-MS21, Castillo-EC17, Hu-MSOM22, Zhang-TRBM21}) have received much attention.


This study is motivated by the fact that profits-as-incentives may promote ridesharing in practice for both MoD operators and drivers.
The potential of ridesharing has been recognized in the academia, but the potential of ridesharing in profit-maximizing platforms/MoDs is not well understood.
The ridesharing problem we study can be summarized in the following (formal definition is given in Section~\ref{sec-preliminary}):
\begin{itemize}
\item A centralized system periodically receives a set of ridesharing offer trips (drivers) and a set of ridesharing request trips (passengers).
We say a passenger is \emph{served} if the passenger is assigned a driver who can deliver the passenger to his/her destination on time. A driver can serve multiple passengers together, which is a ridesharing \emph{match} consisting of a driver and a group of passengers served by the driver. The system computes a \emph{profit} for each match based on the driver's travel distance and time and passengers' itineraries in the match.
An optimizing goal on the profit only, called the \emph{Ridesharing with Profit} (\textbf{RP}) problem, is to maximize the overall profit obtained from the served matches.
In this paper, we focus on a more complex optimization problem, which we call the \emph{Ridesharing with Profit Constraint} (\textbf{RPC}) problem. Instead of maximizing the overall profit, the RPC problem maximizes the number of served 
trips subject to an additional constraint: a minimum profit target specified by the system must be reached.
\end{itemize}
This RPC problem provides a new framework to consider maximizing both the number of passengers served and drivers' profit target.
In particular, the main optimization goal aims at improving the number of passengers served and occupancy rate of vehicles while meeting an overall drivers' profit target.
To the best of our knowledge, such an optimization problem has not been studied before.
Our model allows a flexible pricing for the MoD system operators; and different pricing schemes (e.g., \cite{Li-TRELTR22,Liu-TRCET17,Yan-NRL20}) can be incorporated into our model.
Although the problem studied by Santos and Xavier~\cite{Santos-ESA15} is closely related to the RPC problem, their optimization goals differ from ours, and they focus on heuristics. Similarly, only (meta)heuristics are discussed in~\cite{Hsieh-IJGI20}.

There are two general approaches for handling ridesharing requests in the literature: online and offline approaches.
In the online approach, a request is processed immediately after its arrival (assign to a driver, reject it, or put it in a queue for a small amount of time) without the information of later trips.
In the offline approach, the system accumulates a set of offer and request trips for each time interval (known as \emph{batching}); and the set of trips is processed at once for that interval.
This is a common approach in the literature (e.g.,~\cite{Agatz-TRBM11, Alonso-Mora-PNAS17,Fielbaum-TRCET21,Gu-ISAAC21,Santi-PNAS14,Simonetto-TRCET19}).
An optimization technique, called \emph{rolling horizon}, recomputes the current solution over multiple time intervals, effectively aggregating more trips to be processed at once.
This technique is usually deployed in more technical papers (e.g.,~\cite{Agatz-TRBM11,Najmi-TRELTR17,Nourinejad-TRBM20}).
Under the offline setting, two optimization problems related to the ridesharing problem are the Dial-A-Ride problem (DARP) and the Vehicle Routing problem (VRP).
There are some major differences between these two problems and the RP/RPC problems.
The drivers and passengers in DARP and VRP have less parameters and/or less restricted parameters than that of the drivers and passengers in RPC.
A variant of the Vehicle Routing Problems with Profits, called the Team Orienteering Problem (TOP)~\cite{Archetti-VR14,Gunawan-EJOR16}, is related to RPC problem studied in this paper.
However, profit calculations in VRPPs and TOP are static compared to the profit calculation in RPC, which is more dynamic since it depends on each different driver-passenger(s) assignment.
The most salient difference is that the capacity of a vehicle in DARP and VRP is substantially higher than that in ridesharing. This causes finding an optimal routing for a vehicle in DARP and VRP harder, making many general approaches of DARP and VRP not suitable for the RPC problem since they focus on different fundamentals.

In this paper, we follow the offline setting and our model uses a graph matching approach.
All feasible matches between all drivers and passengers are computed first; and then based on some optimization goal/objective, an assignment consisting of a set of disjoint feasible matches is selected.
Such an approach is shown to be plausible in practice~\cite{Alonso-Mora-PNAS17,Gu-ISAAC21,Simonetto-TRCET19} when traffic delay is minimal.
By this approach, the RP problem can be converted to the maximum weight hypergraph matching (or maximum weight set packing problem), which is NP-hard in general~\cite{Garey79,Karp72}.
We give a mathematical formulation, an exact and approximate algorithms for the RPC problem. Our algorithms are based on applications of maximum matching in bipartite graphs and hypergraphs (set packing). 
We also conduct empirical studies on our algorithms. One hurdle for empirical studies for the RPC problem is the lack of practical data instances. To clear this hurdle, we incorporate the real-world ridesharing dataset from Chicago City with the driver's profit model of Uber to generate test instances for practical scenarios.
Our contributions in this paper are summarized as follows:
\begin{enumerate}\setlength\itemsep{0em}
\item A new optimization problem (the RPC problem) is studied, and a mathematical formulation of the RPC problem is given.
The NP-hardness of the RP problem implies that the RPC problem is NP-hard.

\item We give a polynomial-time exact algorithm framework (including two practical implementations of the algorithm) and a $\frac{1}{2}$-approximation algorithm for a special case of the RPC problem that each match contains $\lambda = 1$ passenger (labeled as RPC1).

\item Another special case of the RPC problem is that only matches with non-negative profit are considered and each match has at most $\lambda\geq 2$ passengers (labeled as RPC\texttt{+}).
This case is still NP-hard, and we give a $\frac{2}{3\lambda}$-approximation algorithm for a specific range of profit target in this case.

\item Based on a real-world ridesharing dataset in Chicago City, profit model of Uber and practical scenarios, we create datasets for an extensive computational study on RPC1 and RPC\texttt{+} problems. Experiment results show that practical profit schemes can be incorporated into our model.
The exact algorithm implementations are efficient, the $\frac{1}{2}$-approximation algorithm achieves 96\% to 99\% (for different practical scenarios), and the $\frac{2}{3\lambda}$-approximation algorithm achieves 90\% of optimal solutions in number of passengers served.
\end{enumerate}

The rest of the paper is organized as follows.
In Section~\ref{sec-preliminary}, we give the preliminaries of the paper and formally define the RPC problem.
In Section~\ref{sec-rpc1-algs}, we describe the exact and approximation algorithms for RPC1.
The $\frac{2}{3\lambda}$-approximation algorithm for RPC\texttt{+} is presented in Section~\ref{sec-app-nonnegative}.
We discuss our numerical experiments and results in Section~\ref{sec-experiment}.
Finally, Section~\ref{sec-conclusion} concludes the paper.

\section{Preliminaries}\label{sec-preliminary}
Let $G(V,E,w)$ be an \emph{edge-weighted} graph with $w:E\rightarrow \mathbb{R}$ assigns each edge $e\in E$ a weight $w(e)$.
A path $P$ in $G$ is a sequence of vertices $v_1,v_2,\ldots,v_p$ such that $(v_i, v_{i+1})$ is an edge of $E$ for $1\leq i \leq p-1$ and denoted by $P=(v_1,\ldots,v_p)$.
The \emph{distance} of a path $P=(v_1,\ldots,v_p)$ is defined as $\dist(P)=w(P)=\sum_{i=1}^{p-1}w(v_i,v_{i+1})$.
The \emph{length} of a path $P$ is the number of edges in $P$, denoted by $|P|$.
A cycle is even (odd) if it has even (odd) length.
A cycle is called \emph{negative} if the sum of the weights of edges in the cycle is negative.

An MoD system has a road network, modeled as a directed graph $G(V,E,w)$, where $V$ is the set of vertices representing geographical sites, 
$E\subseteq V\times V$ is the set of edges (each edge represents a connection between two sites), and a distance function $w: E\rightarrow \mathbb{R}$ that assigns each edge a weight.
The system periodically receives two sets of trips: a set $D=\{\eta_1,\ldots,\eta_k\}$ of $k$ drivers (each operates a vehicle) and a set $R=\{r_1,\ldots,r_l\}$ of $l$ passengers.
Each driver $\eta_i \in D$ is represented by a tuple $(o_i,d_i,\lambda_i)$ of parameters containing 
an origin location $o_i$ (a vertex $o_i \in V(G)$), a destination location $d_i \in V(G)$, 
and a passenger capacity $\lambda_i\geq 1$ of $\eta_i$'s vehicle.
Each driver $\eta_i$ also has an earliest departure time at $o_i$, a latest arrival time at $d_i$, a detour time/distance limit and a maximum trip duration.
A personal driver and a designated driver only differ in the range of some parameters. For example, a designated driver may have a more flexible schedule time and detour limit known by the system.
Each passenger $r_i \in R$ is represented by a tuple $(o_i,d_i)$ as defined above, along with an earliest departure time at $o_i$, a latest arrival time at $d_i$ and a maximum trip duration.

For a driver $\eta_i \in D$ and a group of passengers $R_i \subseteq R$, 
$(\eta_i,R_i)$ is a \emph{feasible match} if there exists a feasible path $\FP(\eta_i,R_i)$ in $N$ used by $\eta_i$ to deliver all of $R_i$ such that travelling along $\FP(\eta_i,R_i)$ satisfies all constraints specified by $\eta_i$ and every $r_j \in R_i$.
These constraints include $|R_i| \leq \lambda_i$, detour limit of $\eta_i$, and time constraints of $\eta_i$ and $R_i$.
An \emph{assignment} $\Pi$ is a set of feasible matches such that for every two feasible matches $(\eta_i, R_i)$ and $(\eta_j,R_j)$ in $\Pi$, $\eta_i\neq \eta_j$ and $R_i \cap R_j = \emptyset$.
For an assignment $\Pi= \{(\eta_i,R_i) \mid \eta_i\in D, R_i\subseteq R\}$, each driver $\eta_i$ in the assignment follows a \emph{shortest feasible path} $\SFP(\eta_i,R_i)$ to serve trips in $R_i$. 
For example, let $R_i=\{r_a,r_q\}$ be the set of passengers in feasible match $(\eta_i, R_i)$.
There are six different visiting orders of $R_i$ in which the passengers of $R_i$ can be picked-up and dropped-off by $\eta_i$, which correspond to six paths in road network $G(V,E,w)$ (each starts at $o_i$ and ends at $d_i$).
Below are the six visiting orders of $R_i$ for driver $\eta_i$:
\begin{align*}
\{&(o_i,o_a,o_q,d_a,d_q,d_i), (o_i,o_a,d_a,o_q,d_q,d_i),(o_i,o_a,o_q,d_q,d_a,d_i),\\
&(o_i,o_q,o_a,d_a,d_q,d_i), (o_i,o_q,d_q,o_a,d_a,d_i),(o_i,o_q,o_a,d_q,d_a,d_i)\}.
\end{align*}%
Path $\SFP(\eta_i,R_i)$ is the path in $G$ corresponds to one of the six visiting orders that is feasible and has the shortest distance.
Every feasible match $(\eta_i,R_i)$ along with $\SFP(\eta_i,R_i)$ can be computed efficiently with small $|R_i|$, as described in~\cite{Alonso-Mora-PNAS17,Gu-ISAAC21,Simonetto-TRCET19}.

Each feasible match $(\eta_i,R_i)$ is associated with a revenue $rev(\eta_i,R_i)$, a travel cost $tc(\eta_i,R_i)$ and a profit $w(\eta_i,R_i)$, which are computed by the MoD system.
Most or all of $rev(\eta_i,R_i)$ are given to the driver $\eta_i$ for serving all of $R_i$.
The revenue $rev(\eta_i,R_i)$ is assumed to be computed based on $\FP(\eta_i,R_i)$ and should be at most what the passengers of $R_i$ pay to the MoD system.
The travel cost $tc(\eta_i,R_i)$ for $\eta_i$ accounts for traversing $\SFP(\eta_i,R_i)$.
The revenue and cost are decided by several parameters such as $\dist(\SFP(\eta_i,R_i))$, travel time, regions, pricing policies, etc.
For example, a simple revenue $rev(\eta_i,R_i)$ can just be $\sigma\cdot \dist(\SFP(\eta_i,R_i))$ for some constant $\sigma$.
The profit of a feasible match $(\eta_i,R_i)$ is $w(\eta_i,R_i)=rev(\eta_i,R_i) - tc(\eta_i,R_i)$, which can be negative, and we assume it is expressed in integers (e.g., cents, smallest payable amount).
In our computational study, we estimate $rev(\eta_i,R_i)$ and $tc(\eta_i,R_i)$ based on the profit model of Uber and practical scenarios, as described in Section~\ref{sec-experiment}.
The revenue and cost factors may vary over time/by region and also be different for each driver and feasible match, following some pricing strategies chosen by the system operators~\cite{Li-TRELTR22,Liu-TRCET17}.

The RP (ridesharing with profit) problem is to assign passengers of $R$ to drivers $D$ with overall profit maximized. The RP problem can be formulated as follows.
\begin{small}
\begin{alignat}{4}
& \max_{\Pi}   &       & \sum_{(\eta_i,R_i)\in \Pi} w(\eta_i,R_i) & \qquad \tag{\romannumeral1} \label{formulation-sharingmodel}\\
& \text{subject to } & \qquad & \eta_i\neq \eta_j \wedge R_i\cap R_j=\emptyset,  & & \forall (\eta_i,R_i),(\eta_j,R_j)\in \Pi : (\eta_i,R_i)\neq(\eta_j,R_j) \tag{\romannumeral2} \label{constraint-disjoint-RP}
\end{alignat}
\end{small}%
The objective function~\eqref{formulation-sharingmodel} is to maximize the overall profit obtained from served trips. Constraint~\eqref{constraint-disjoint-RP} ensures that each passenger request is assigned to only one driver and each driver serves at most one feasible match (a unique group of passengers).
Note that not all requests of $R$ are required to be served in an assignment $\Pi$.

In this paper, we focus on a more complex optimization problem, called the \emph{Ridesharing with Profit Constraint} (RPC) problem.
In application, MoD may want to serve as many passengers as possible while maintaining a \emph{profit target} $c$.
With this in mind, we introduce a profit constraint and a formulation for the RPC problem as follows.
\begin{small}
\begin{alignat}{4}
& \max_{\Pi}   &        & \sum_{(\eta_i,R_i)\in \Pi} |R_i| & \qquad \tag{\romannumeral3} \label{formulation-max-passenger}\\
& \text{subject to } & \qquad & \eta_i\neq \eta_j \wedge R_i \cap R_j = \emptyset,  & & \forall (\eta_i,R_i),(\eta_j,R_j)\in \Pi : (\eta_i,R_i)\neq(\eta_j,R_j) \tag{\romannumeral4} \label{constraint-disjoint}\\
&                    &        & \sum_{(\eta_i,R_i)\in \Pi} w(\eta_i,R_i) \geq c & & \tag{\romannumeral5} \label{constraint-profit}
\end{alignat}
\end{small}%
The objective function~\eqref{formulation-max-passenger} is to maximize the total number of passengers served. Constraint~\eqref{constraint-disjoint} is the same as constraint~\eqref{constraint-disjoint-RP}.
Constraint~\eqref{constraint-profit} ensures the system profit meets a given target.
An assignment $\Pi$ containing any feasible match $(\eta_i,R_i)$ with negative profit ($w(\eta_i,R_i)<0$) means that the driver $\eta_i$ loses money.
In practice, the MoD system operator could choose to compensate such a driver $\eta_i$ in hope to increase market share.

We construct an integer-weighted hypergraph $H(V,E,w)$ to represent the formulation \eqref{formulation-max-passenger}-\eqref{constraint-profit} as follows.
Initially, $V(H) = D \cup R$.
For each $\eta_i \in D$ and for every subset $R_i$ of $R$ with $1 \leq |R_i| \leq \lambda_i$, create a hyperedge $e=\{\eta_i\} \cup R_i$ in $E(H)$ if $(\eta_i, R_i)$ is a feasible match.
Each edge $e=\{\eta_i\} \cup R_i \in E(H)$ has weight $w(e)=w(\eta_i,R_i)$, the profit of $\eta_i$.
Remove all isolated vertices from $H$.
Let $H^-$ be the subgraph of $H$ such that $H^-$ contains all edges of $H$ with negative weight and $H^+ = H \setminus H^-$.
For an assignment $\Pi$, let $w(\Pi) = \sum_{(\eta_i,R_i)\in \Pi} w(\eta_i,R_i)$ be the profit of $\Pi$.
There are at most $\sum_{1 \leq a \leq \lambda_i} \binom{l}{a}$ edges incident to each $\eta_i$ in $H$.
Let $\lambda = \max_{\eta_i \in D} \lambda_i$.
If $\lambda$ is a small constant, the size of $H$ is polynomially bounded.
In practice, it is reasonable to assume $\lambda$ is small; however when $\lambda$ is not small, we may purposely restrict the number of edges incident to each vertex so that $|E(H)|$ becomes reasonable for practice.

For an edge-weighted (hyper)graph $G(V,E)$ and $E' \subseteq E(G)$, the weight of $E'$ is denoted by $w(E') = \sum_{e \in E'} w(e)$, where $w(e)$ is the weight of edge $e$.
A \emph{matching} $M$ in a (hyper)graph $G$ is a set of edges of $G$ such that every pair of edges in $M$ do not have a common vertex.
The size $|M|$ of a matching $M$ is the number of edges in $M$ and the weight of $M$ is $w(M)$.
The RPC problem is then to find a matching $M$ in $H$ such that $\sum_{\{\eta_i\} \cup R_i \in M} |R_i|$ is maximized and $w(M)>c$.

Let $M_1$ and $M_2$ be two matchings in a (hyper)graph $G(V,E)$.
Denoted by $F(V,E) = M_1 \Delta M_2$ is the resulting graph of the symmetric difference of $M_1$ and $M_2$.
Let $\mathcal{F}$ be the set of connected components in $F$.
It is well known that when $G$ is a graph, each component of $\mathcal{F}$ is either a path or an even cycle where the edges are alternating between $M_1$ and $M_2$.
For any subset $\mathcal{F}' \subseteq \mathcal{F}$, let $E(\mathcal{F}')=\cup_{C\in \mathcal{F}'} E(C)$ and $w(\mathcal{F}')= \sum_{C\in \mathcal{F}'} w(E(C))$.

Let $c^*$ be the weight of a maximum weight matching in the above constructed hypergraph.
It is easy to see that finding a maximum weight matching $M^*$ in $H^+$ solves the formulation~\eqref{formulation-sharingmodel}-\eqref{constraint-disjoint-RP} and vice versa.
Since formulation~\eqref{formulation-sharingmodel}-\eqref{constraint-disjoint-RP} is a weighted $(\lambda+1)$-set packing formulation, finding matching $M^*$ and $c^*$ (RP problem) are NP-hard in general for $\lambda \geq 2$~\cite{Garey79,Karp72}.

\begin{theorem}
The RPC problem is NP-hard for an arbitrary target $c$ and $\lambda\geq 2$.
\label{theorem-RPC-NPhard}
\end{theorem}

\begin{proof}
Given the formulation~\eqref{formulation-sharingmodel}-\eqref{constraint-disjoint-RP} for an instance of the RP problem, construct an instance of the RPC problem containing the same set of feasible matches $\cup_{\eta_i \in D, R_i\subseteq R} (\eta_i, R_i)$.
For a driver $\eta_i \in D$, let $\omega^*_i$ be the profit of the match containing $\eta_i$ with the largest profit.
Set the profit target $c = \sum_{\eta_i\in D} \omega^*_i$ for the constructed RPC problem instance.
Note that a driver $\eta_i$ cannot appear in more than one match in any solution of the RP and RPC problems.
Hence, the RPC problem has a feasible solution if and only if the objective function value of formulation~\eqref{formulation-sharingmodel}-\eqref{constraint-disjoint-RP} is equal to $c$.
Since solving the formulation~\eqref{formulation-sharingmodel}-\eqref{constraint-disjoint-RP} is NP-hard for $\lambda \geq 2$, the RPC problem is NP-hard for $\lambda \geq 2$.
\end{proof}

Further, Hazan et al.~\cite{Hazan-CC06} showed that the $(\lambda+1)$-set packing problem cannot be approximated to within $\Omega(\frac{\text{ln} (\lambda+1)}{\lambda+1})$ in general for $\lambda\geq 2$.
There exists a polynomial-time $\frac{2}{\lambda+2}$-approximation algorithm for approximating the maximum profit of $\Pi$~\cite{Berman-SWAT00}.
However, similar algorithms~\cite{Berman-SWAT00, Chandra-JoA01} cannot be directly applied to the RPC problem since these algorithms only approximate the maximum profit $w(\Pi)$ and do not consider the size of each subset (match) in $\Pi$ and the different elements in the subset/match.
Algorithms for the maximum set packing problem (e.g., \cite{Furer-ISCO14,Sviridenko-ICALP13}) cannot apply to the RPC problem either since such algorithms do not consider general integer weight.
Due to the NP-hardness of the RPC problem (Theorem~\ref{theorem-RPC-NPhard}) and the inapproximability of the weighted set packing problem, we study two variants of the RPC problem: RPC1 and RPC\texttt{+}.
The RPC1 problem variant assumes that for a given instance of the RPC problem, $\lambda_i = 1$ for every driver $\eta_i \in D$ ($\lambda=1$).
To solve the RPC1 variant, we use an approach in solving the maximum matching problem on bipartite graphs.
For the RPC\texttt{+} problem variant, we include one more constraint (called the \textit{non-negative profit constraint}) to formulation~\eqref{formulation-max-passenger}-\eqref{constraint-profit} of the RPC problem:
\[
w(\eta_i,R_i) \geq 0, \forall (\eta_i,R_i) \in \Pi.  \tag{\romannumeral6} \label{constraint-nonnegative-edge}
\]
To solve the RPC\texttt{+} variant, we use a local search approach similar to the ones in~\cite{Berman-SWAT00, Chandra-JoA01}.

\section{RPC1 variant} \label{sec-rpc1-algs}
For $\lambda = 1$, the weighted hypergraph $H(V,E,w)$, constructed in Section~\ref{sec-preliminary}, becomes a weighted bipartite graph. 
A solution to the RPC problem for $\lambda=1$ (with $c \leq c^*$) is a matching $M$ in $H$ with $w(M)\geq c$ and $|M|$ maximized.
We first give a polynomial-time exact algorithm (referred to as \textbf{ExactNF}) framework that uses network flow to find an optimal solution.
Then, we describe two implementations of ExactNF that are suitable for practice.
Finally, we give a simple $\frac{1}{2}$-approximation algorithm (referred to as \textbf{Greedy}).

\subsection{Exact algorithm}
The framework description of ExactNF is given below.
\begin{enumerate}\setlength\itemsep{0em}
\item Construct a flow network $N(V,E)$ from $H$, where $V(N) = \{s,t\} \cup V(H)$, $s$ is the source, and $t$ is the sink.
For each $\eta_i\in V(H)$, create an edge $(s,\eta_i)$ in $E(N)$ with cost $0$ and capacity $1$. 
For each $(\eta_i,r_j) \in E(H)$, create an edge $(\eta_i,r_j)$ in $E(N)$ with cost $-w(\eta_i,r_j)$ and capacity $1$.
For each $r_j\in V(H)$, create an edge $(r_j,t)$ in $E(N)$ with cost $0$ and capacity $1$.
Note that the maximum amount of flow that can be sent from $s$ to $t$ in $N$ is at most $n_{min}=\min\{|V(H) \cap D|, |V(H) \cap R|\}$.

\item For $1 \leq y\leq n_{min}$, find a minimum cost flow $f_y$ of value $y$ (sent from $s$ to $t$) or conclude that there is no flow of value $y$ in $N$.

\item For an edge $e \in E(N)$, let $f_y(e)$ be the flow value passing through $e$ in $f_y$. Let $c(f_y) = \sum_{e\in E(N) \mid f_y(e) > 0} w(e)$ be the cost of flow $f_y$.
If a flow $f_y$ with $c(f_y)\leq -c$ is computed in Step 2, then $y=\text{argmax}_{y} -c(f_y)\geq c$, and output the edges $\cup_{e\in E(N)\mid f_y(e)>0 \wedge e\in E(H)}$ with positive flow value in $f_y$ as solution $M$; otherwise, conclude there is no matching in $H$ with profit at least $c$.
\end{enumerate}

\begin{theorem}\label{theorem-exact-alg}
Algorithm ExactNF finds a matching $M$ with $w(M)\geq c$ and $|M|$ maximized or concludes that there is no matching $M$ with $w(M)\geq c$ in $H$ in polynomial time.
\end{theorem}

\begin{proof}
The edges $\cup_{e\in E(N)\mid f_y(e)>0 \wedge e\in E(H)}$ of a flow $f_y$ form a matching $M$ in $H$ of cardinality $y$.
Since the cost $c(f_y)$ is minimum among all flows of value $y$ and the profit $w(M)$ is the negation of $c(f_y)$, $w(M)$ is maximum of all matchings in $H$ of cardinality $y$.
If there is a matching $M$ in $H$ with $w(M)\geq c$, then $1\leq |M|\leq n_{min}$ (since $\lambda=1$) and Algorithm ExactNF finds the matching $M$ of the largest cardinality with $w(M)\geq c$.
An upper bound on the running time of ExactNF is $O(n_{min}\cdot t(N))$, where $t(N)$ is the time to compute a min-cost flow $f_y$ and is polynomial in the size of $N$~\cite{Ahuja-NF93}.
\end{proof}

The computational time of $t(N)$ heavily depends on how $f_y$ is computed.
We give two practical implementations of ExactNF (referred to as Algorithm \textbf{ExactNF1} and Algorithm \textbf{ExactNF2}) to compute $f_y$.
The first one uses a linear programming (LP) approach to find $f_y$ by a min-cost flow LP formulation, and the second one to find $f_y$ by graph algorithms.
Algorithm \textbf{ExactNF1} is described in the following.
\begin{enumerate}
\item Let $N'$ be the network $N(V,E)$ constructed above without the edge costs.
Construct a maximum flow LP formulation for $N'$ and solve the LP.
Let $y^*$ be the maximum flow value (optimal solution of the LP).

\item For $y=y^*$ to 1, find a minimum cost flow of value $y$ in $N$ (sent from $s$ to $t$) by solving another LP formulation for a min-cost flow in $N$.
If $c(f_y)\leq -c$, stop and output $\cup_{e\in E(N)\mid f_y(e)>0 \wedge \{u,v\}\in E(H)}$. Otherwise, continue until $c(f_y)\leq -c$ or conclude there is no such flow in $N$, implying $H$ does not have a feasible solution.
\end{enumerate}
In the worst case, such a simple implementation may need to solve the min-cost flow at most $n_{min}$ times.
However, most (majority) of the edges in $H$ have positive weight (negative weight in $N$) in practical scenarios, implying the minimum cost flow $f_y$ that satisfies $c(f_y)\leq -c$ usually has $y$ close to $y^*$.
Further, if $c$ is smaller than the largest profit obtainable for a noticeable amount (say 20\%), $y$ is very close to $y^*$.
These suggest that the number of min-cost flows to be computed is a small constant in practice.
Due to the restricted search space provided by the given upper flow value $y$, computing a min-cost flow $f_y$ on $N$ is faster than that on the original network $N$.
An efficient implementation for ExactNF1 is to compute $f_y$ for $y$ from $y^*$ down to 1 and output the solution when the first $f_y$ with $c(f_y)\leq -c$ is found.
From Theorem~\ref{theorem-exact-alg}, we have the following corollary.

\begin{corollary}\label{corollary-exact1-alg}
Algorithm ExactNF1 finds a matching $M$ with $w(M)\geq c$ and $|M|$ maximized or concludes that there is no matching $M$ with $w(M)\geq c$ in $H$ in $O(n_{min}\cdot t(N))$ time, where $t(N)$ is the time to find a min-cost flow $f_y$ by an LP solver.
\end{corollary}

The second approach for computing $f_y$ is by graph algorithms described in~\cite{Ahuja-NF93}: 
Let $f_y$ be the min-cost flow of flow value $y$ in $N$ and $N_{f_y}$ be the residual network of $N$ w.r.t. $f_y$, where $N_{f_0}=N$.
We compute $f_y$ in $N_{f_{y-1}}$ w.r.t. $f_{y-1}$ for $y=1,2,\ldots,n_{min}$.
As shown in~\cite{Ahuja-NF93}, the min-cost flow $f_y$ can be computed by the successive shortest path algorithm (SSPA).
For completeness and implementation detail, we give the details of the SSPA, including how
our stopping conditions are applied (which is different from the SSPA).
Further, definitions used in the correctness proof of Algorithm ExactNF2 are also introduced.
The detailed implementation of Algorithm \textbf{ExactNF2} is described below.
\begin{enumerate}
\item If $N(V,E)$ has negative weighted edges, first change the negative edge weights into non-negative weights as in Johnson's shortest path algorithm~\cite{Cormen-IntroAlg09,Johnson-JoA77}:
Use Bellman-Ford algorithm to compute the shortest distance $\dist(s,u)$ for every $u \in V(N)$.
Assign $h(u)=\dist(s,u)$ and compute a new cost $\hat{w}(u,v)=w(u,v)+h(u)-h(v)$ for every
$(u,v)\in E(N)$, then $\hat{w}(u,v)\geq 0$.
Label the network with the new weights as $\hat{N}(V,E)$ (a min-cost flow $f_y$ of value $y$ in $\hat{N}$ is a min-cost flow of value $y$ in $N$).

\item The rest is similar to the successive shortest path algorithm, except the stopping conditions are different.
Initialize the \emph{node potential} $\pi(u)=0$ for each $u\in V(\hat{N})$ and \emph{reduced cost} $w_{\pi}(u,v)=\hat{w}(u,v)-\pi(u)+\pi(v)$ (which is $\hat{w}(u,v)$ initially) for each $(u,v) \in E(\hat{N})$.
Let $f_0$ be an empty initial flow on $\hat{N}$.
Let $\hat{N}_{f_{y-1}}(\pi)$ be the residual network w.r.t. flow $f_{y-1}$ and reduced costs $w_{\pi}$, where $\hat{N}_{f_0}(\pi) = \hat{N}$.

For $y=1,\ldots,n_{min}$, find a minimum cost flow $f_y$ of value $y$ from $f_{y-1}$ as follows.
Compute single-source shortest paths from $s$, $SSSP(s)$, to every other vertex in $u \in V(\hat{N}_{f_{y-1}}(\pi))$ to get $\dist_{y-1}(s,u)$.
Update node potential $\pi(u) = \pi(u) - \dist_{y-1}(s,u)$ for every $u \in V(\hat{N})$.
Update reduced cost for every edge of $\hat{N}_{f_{y-1}}(\pi)$:
$w_{\pi}(u,v)=\hat{w}(u,v)-\pi(u)+\pi(v)$ if $(u,v) \in E(\hat{N}) \cap E(\hat{N}_{f_{y-1}}(\pi))$; otherwise, $w_{\pi}(u,v)$ is unchanged.
Let $P_{y-1}$ be the shortest $s-t$ path found by $SSSP(s)$.
Then, augment flow along $P_{y-1}$ to get a flow $f_y$, and construct the residual network $\hat{N}_{f_y}(\pi)$ w.r.t. $f_y$ and $w_{\pi}$.
Note that each edge $(u,v)$ in $P_{y-1}$ is in $\hat{N}_{f_{y-1}}(\pi) \setminus \hat{N}_{f_y}(\pi)$; and its reduced cost $w_{\pi}(u,v)$ is updated prior to augmenting flow along $P_{y-1}$.
After the augmentation, each such edge $(u,v)$ in $P_{y-1}$ has reduced cost $w_{\pi}(v,u)=-w_{\pi}(u,v)$ in $\hat{N}_{f_y}(\pi)$.
Due to the properties of the successive shortest path algorithm, $w_{\pi}(u,v)\geq 0$ for every $(u,v) \in E(\hat{N}_{f_y}(\pi))$ (reduced cost optimality conditions~\cite{Ahuja-NF93}).

\item Let $f_y$ be the flow found by Step 2 in each iteration. For an edge $(u,v) \in E(\hat{N})$, let $f_y(u,v)$ be the flow value passing through $(u,v)$ w.r.t. $f_y$.
Let $c(f_y)$ be the cost of flow $f_y$ in $N$, namely, $c(f_y) = \sum_{(u,v)\in E(\hat{N}) \mid f_y(u,v) > 0} \hat{w}(u,v) + h(v)- h(u) = \sum_{(u,v)\in E(\hat{N}) \mid f_y(u,v) > 0} w(u,v)$.
Stop Step 2 if:

\begin{itemize}
\item either the value $y \leq n_{min}$ of flow $f_y$ cannot be increased (an $s-t$ path cannot be found in $\hat{N}_{f_y}(\pi)$), or $c(f_{y+1}) > c(f_y)$ such that $c(f_{y+1}) > -c$.
For either case, if $c(f_y)\leq -c$ then output $\cup_{e\in E(\hat{N})\mid f_y(e)>0 \wedge e\in E(H)}$ as a solution. Otherwise, conclude there is no matching in $H$ with profit at least $c$.
\end{itemize}

\end{enumerate}

Johnson's shortest path algorithm has the following properties that also apply to ExactNF2.
\begin{property}\label{property-new-weight}
\cite{Cormen-IntroAlg09,Johnson-JoA77}.
\begin{enumerate}[a$)$]\setlength\itemsep{0em}
\item For every edge $(u,v) \in E(\hat{N})$, $\hat{w}(u,v) \geq 0$.
\item A flow $f_y$ of flow value $y$ is a min-cost flow in $N$ if and only if $f_y$ is a min-cost flow of flow value $y$ in $\hat{N}$.
\end{enumerate}
\end{property}

\begin{proof}
Since there is no cycle in $N(V,E)$, there is no negative cycle in $N(V,E)$.
There are only out-going edges from source $s$ with zero weight.
It follows from the analysis in \cite{Cormen-IntroAlg09} that part (a) holds;
and all edges of $E(N)$ can be changed to non-negative weighted edges $E(\hat{N})$ by using the Bellman-Ford algorithm.

Part (b), which is an extension of a property in~\cite{Cormen-IntroAlg09} on the length of a path in $N$ and $\hat{N}$.
Let $Q'= (s=v_0,v_1,\ldots,v_{z-1},v_z=t)$ be an $s-t$ path in $\hat{N}$.
As shown in~\cite{Cormen-IntroAlg09}, the weight of $Q'$ in $\hat{N}$ is
\begin{align} \label{eq-optimal-flow}
\hat{w}(Q')=\sum_{1\leq j\leq z} \hat{w}(v_{j-1},v_j)&=\sum_{1\leq j \leq z} w(v_{j-1},v_j)+h(v_{j-1})-h(v_j) \\
&=h(s)-h(t)+\sum_{1\leq j \leq z} w(v_{j-1},v_j) \nonumber \\
&=h(s)-h(t)+ w(Q'), \nonumber
\end{align}
where $w(Q')$ is the weight of $Q'$ in $N$.
For any flow $f_y$ in $N$ or $\hat{N}$, the edges with positive flow of $f_y$ form a set $\mathcal{Q} =\{Q_1,\ldots, Q_y\}$ of edge-disjoint $s-t$ paths since each edge in $N$ and $\hat{N}$ has unit capacity.
From Eq~\eqref{eq-optimal-flow}, the weight of $\mathcal{Q}$ in $\hat{N}$ is
\begin{align}
\hat{w}(\mathcal{Q}) = \sum_{1\leq i\leq y} \hat{w}(Q_i)=y(h(s) - h(t)) + \sum_{1\leq i\leq y} w(Q_i).
\end{align}
Since $h(s)$ and $h(t)$ are independent of any $s-t$ path, if $f_y$ is a min-cost flow ($\sum_{1\leq i\leq y} w(Q)$ is minimum) in $N$, then $f_y$ is a min-cost flow in $\hat{N}$ and vice versa.
\end{proof}

The SSPA has the following properties that also apply to ExactNF2 since the initial network $\hat{N}$ has non-negative edge weights and $SSSP(s)$ on $\hat{N}$ can be computed correctly by Property~\ref{property-new-weight}.
\begin{property}\label{property-reduced-cost-optimality}
\cite{Ahuja-NF93}
\begin{enumerate}[a$)$]\setlength\itemsep{0em}
\item (Reduced cost optimality conditions) For $y\geq 0$ and $\hat{N}_{f_y}(\pi)$, a feasible flow $f_y$ is an optimal solution to the min-cost flow problem if and only if some set of node potentials $\pi$ satisfy that $w_{\pi}(u,v)\geq 0$ for every edge $(u,v)\in\hat{N}_{f_y}(\pi)$.
\item Sending flow along an $s-t$ shortest path $p_y$ in $\hat{N}_{f_y}(\pi)$ w.r.t. reduced costs $w_{\pi}$ still maintains the reduced cost optimality conditions in each iteration.
\end{enumerate}
\end{property}

Next, we give a general definition of the weight of a path in $\hat{N}_{f_i}(\pi)$ for $N$ and $\hat{N}$.
Let $\hat{N}_{f_i}(\pi)$ be a residual network w.r.t. flow $f_i$ ($0 \leq i \leq n_{min}$) and reduced cost $w_{\pi}$, where $\hat{N}_{f_0}(\pi) = \hat{N}$.
For any path $P$ in $\hat{N}_{f_i}(\pi)$, define the weight of $P$ in $N$ as 
\[
w(P) = \sum_{(u,v)\in P\cap E(N)} w(u,v) - \sum_{(u,v)\in P\setminus E(N)} w(v,u)\]
and the weight of $P$ in $\hat{N}$ as 
\[
\hat{w}(P) = \sum_{(u,v)\in P\cap E(\hat{N})} \hat{w}(u,v) - \sum_{(u,v)\in P\setminus E(\hat{N})} \hat{w}(v,u).
\]
Note that if $P\setminus E(N) = \emptyset$, the weight $w(P)$ is the regular definition of $w(P)$ (same for $\hat{w}(P)$ as $E(N)=E(\hat{N})$).

\begin{corollary}\label{corollary-general-weight-path}
Let $P$ be an $s-t$ path in the residual network $\hat{N}_{f_i}(\pi)$ w.r.t. flow $f_i$ and reduced cost $w_{\pi}$.
The weight of $P$ in $\hat{N}$ is $\hat{w}(P) = h(s) - h(t) + w(P)$.
\end{corollary}

\begin{proof}
If $P\setminus E(N) = \emptyset$, then from Eq~\eqref{eq-optimal-flow}, corollary holds.
Suppose $P \setminus E(N) \neq \emptyset$.
Then, the weight of $P$ in $\hat{N}$ is
\begin{align*}
\hat{w}(P)&=\sum_{(u,v)\in P\cap E(\hat{N})} \hat{w}(u,v) - \sum_{(u,v)\in P\setminus E(\hat{N})} \hat{w}(v,u) \\ &=[\sum_{(u,v)\in P\cap E(N)} w(u,v)+h(u)-h(v)] - [\sum_{(u,v)\in P\setminus E(N)} w(v,u)+h(v)-h(u)] \\
&=[\sum_{(u,v)\in P\cap E(N)} w(u,v)+h(u)-h(v)] + [\sum_{(u,v)\in P\setminus E(N)} h(u) -h(v)] - \sum_{(u,v)\in P\setminus E(N)} w(v,u) \\
&=h(s)-h(t)+\sum_{(u,v)\in P\cap E(N)} w(u,v) - \sum_{(u,v)\in P\setminus E(N)} w(v,u) \\
&=h(s)-h(t)+ w(P).
\end{align*}
Hence, we have the corollary.
\end{proof}

\begin{lemma}\label{lemma-concave-flow-cost}
Let $f_y$ and $f_{y+1}$ be two flows found in Step 2 of ExactNF2.
If $c(f_y) < c(f_{y+1})$, then $c(f_{y+1}) \leq c(f_z)$ for every flow $f_z$ found after $f_{y+1}$, $z> y+1$.
\end{lemma}

\begin{proof}
Assume for contradiction that there is a flow $f_{z}$ found after $f_{y+1}$ such that $c(f_z) < c(f_{y+1})$ for the smallest $z>y+1$.
Let $P_y$ be the path found in $\hat{N}_{f_y}(\pi)$ to get $f_{y+1}$ from $f_y$ (by augmenting flow along $P_y$), and similarly, let $P_{z-1}=(s, v_1, v_2,\ldots, v_q, t)$ be the path found in $\hat{N}_{f_{z-1}}(\pi)$.
Then, $c(f_{y+1})-c(f_y)=w(P_y) > 0$ and $c(f_z)-c(f_{z-1})=w(P_{z-1}) < 0$.
By Corollary~\ref{corollary-general-weight-path}, $\hat{w}(P_{z-1}) < \hat{w}(P_{y})$.
If $P_{z-1}$ exists in $\hat{N}_{f_y}(\pi)$, by Property~\ref{property-reduced-cost-optimality}, the algorithm would have augmented flow along $P_{z-1}$ instead of $P_y$ since it gives a flow $f'_{y+1}$ of value $y+1$ in $\hat{N}$ with a cost lower than that of $f_{y+1}$.
Assume $P_{z-1}$ does not exist in $\hat{N}_{f_y}(\pi)$.
Let $g \geq y+1$ be the smallest iteration such that $P_{z-1}$ exists in $\hat{N}_{f_g}(\pi)$ ($g\leq z-1$).
Let $P_{g-1}=(s,u_1,u_2,\ldots, u_x, t)$ be the path found in $\hat{N}_{f_{g-1}}(\pi)$.
Some edges of $P_{z-1}$ must have inverse directions in $P_{g-1}$.
Consider any maximal subpath $(v_i,v_{i+1},\ldots,v_j)$ of $P_{z-1}$ such that the reverse path $(v_j,\ldots,v_{i+1},v_i)$ is in $P_{g-1}$.
Label $P_{z-1}$ as three subpaths: $P_{z-1}(1) = (s,v_1,\ldots, v_i)$, $P_{z-1}(2) = (v_i,\ldots, v_j)$ and $P_{z-1}(3) = (v_j, \ldots, v_q,t)$.
Label $P_{g-1}$ as three subpaths: $P_{g-1}(1) = (s,u_1,\ldots, u_a,v_j)$, $P_{g-1}(2) = (v_j,\ldots, v_i)$ and $P_{g-1}(3) = (v_i, u_b,\ldots, u_x,t)$.
Then, $P'_{g-1}=P_{z-1}(1) \cup P_{g-1}(3)$ and $P''_{g-1}=P_{g-1}(1) \cup P_{z-1}(3)$ are two $s-t$ paths in $\hat{N}_{f_{g-1}}(\pi)$ (see Figure~\ref{fig-subpaths} for an example).
\begin{figure}[!ht]
\centering
\includegraphics[width=0.7\linewidth]{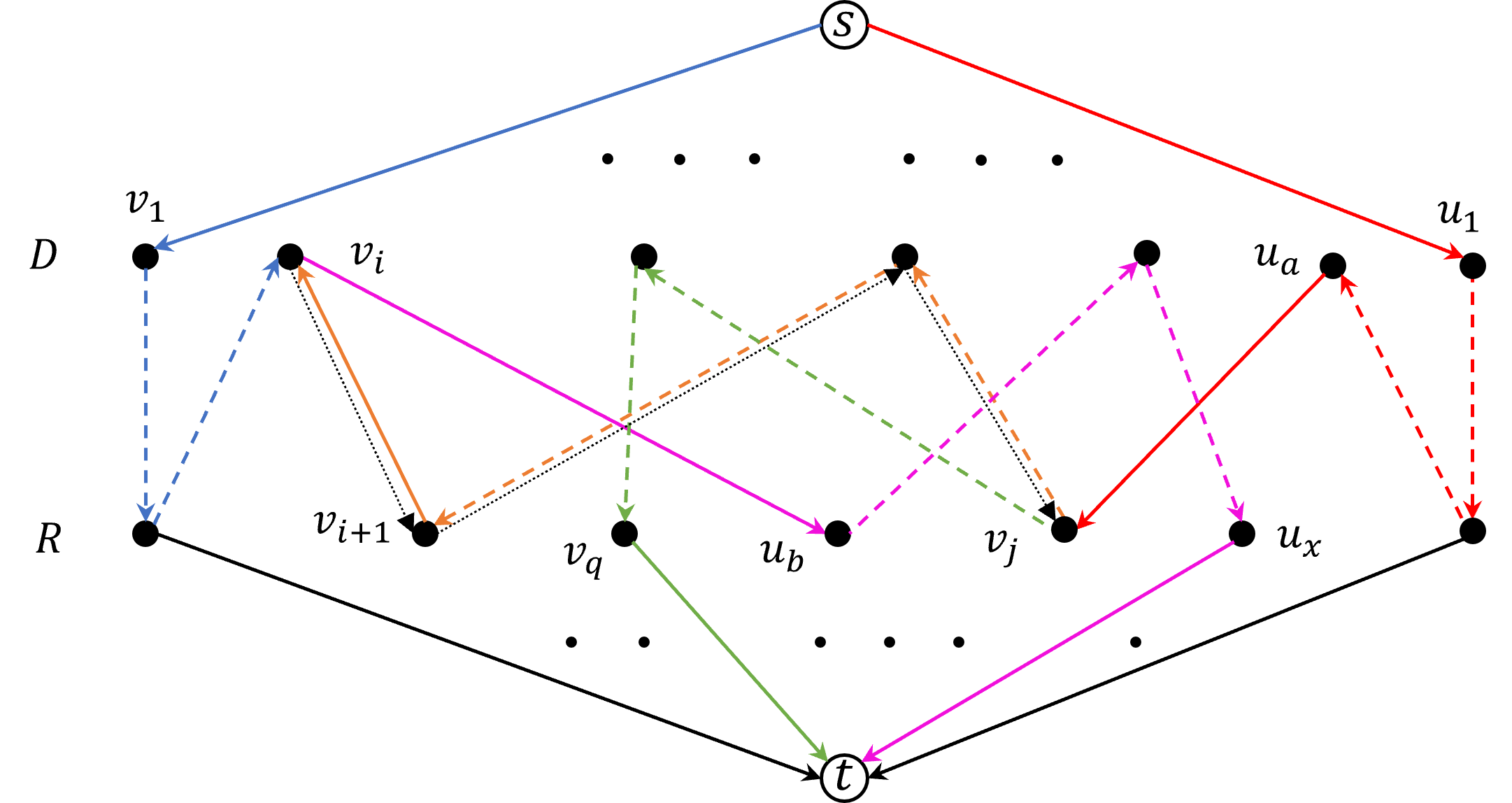}
\captionsetup{font=small}
\caption{Residual network $\hat{N}_{f_{g-1}}(\pi)$. Path $P_{g-1}=(s,u_1,\ldots, u_x, t)$ found in $\hat{N}_{f_{g-1}}(\pi)$ by the algorithm is labeled as three subpaths $P_{g-1}(1)$ (red lines), $P_{g-1}(2)$ (orange lines) and $P_{g-1}(3)$ (pink lines).
The path $P_{z-1}=(s, v_1,\ldots, v_q, t)$ exists in $\hat{N}_{f_g}(\pi)$ and is labeled as three subpaths $P_{z-1}(1)$ (blue lines), $P_{z-1}(2)$ (black dotted lines, these edges are in $\hat{N}_{f_g}(\pi)$ and not in $\hat{N}_{f_{g-1}}(\pi)$) and $P_{z-1}(3)$ (green lines).}
\label{fig-subpaths}
\end{figure}
Recall that $z>y+1$ is the smallest such that $c(f_z) < c(f_{y+1})$, which means $w(P_{g-1}) \geq 0$ and $w(P_{z-1}) < 0$.
Next, we show $w(P'_{g-1}) < w(P_{g-1})$ or $w(P''_{g-1}) < w(P_{g-1})$.
If $w(P_{z-1}(3)) < w(P_{g-1}(2)) + w(P_{g-1}(3))$, then 
\[
w(P''_{g-1}) = w(P_{g-1}(1)+ w(P_{z-1}(3)) < w(P_{g-1}(1)) + w(P_{g-1}(2)) + w(P_{g-1}(3)) = w(P_{g-1}).
\]
Suppose $w(P_{z-1}(3)) \geq w(P_{g-1}(2)) + w(P_{g-1}(3))$.
Note that $P_{z-1}(2)$ is the reversed path of $P_{g-1}(2)$, and $w(P_{z-1}(2)) = -w(P_{g-1}(2))$ 
since for each $(u,v)$ in $P_{z-1}(2)$, $w(u,v) = -w(v,u)$, where $(v,u)$ is in $P_{g-1}(2)$.
Because $w(P_{z-1}) < 0$, $w(P_{z-1}(1)) + w(P_{z-1}(3)) < -w(P_{z-1}(2))$, so
\[
w(P_{g-1}(2)) > w(P_{z-1}(1)) + w(P_{z-1}(3)).
\]
From this,
\begin{align*}
w(P_{z-1}(3)) &\geq w(P_{g-1}(2)) + w(P_{g-1}(3)) \\
&> w(P_{z-1}(1)) + w(P_{z-1}(3)) + w(P_{g-1}(3)),
\end{align*}
implying $w(P'_{g-1}) = w(P_{z-1}(1)) + w(P_{g-1}(3)) < 0$.
Thus, $w(P'_{g-1}) < w(P_{g-1})$.
From the above and Corollary~\ref{corollary-general-weight-path}, at least one of $\hat{w}(P'_{g-1}) < \hat{w}(P_{g-1})$ and $\hat{w}(P''_{g-1}) < \hat{w}(P_{g-1})$ is true.
This implies that we can obtain a min-cost flow $f'_g$ of value $g$ in $\hat{N}$ by augmenting flow along $P'_{g-1}$ (or $P''_{g-1}$) instead of $P_{g-1}$ such that $f'_g$ has a cost lower than that of $f_g$.
This is a contradiction to Property~\ref{property-reduced-cost-optimality} that $f_g$ is a min-cost flow of value $g$.
Therefore, a flow $f_z$ with $c(f_z) < c(f_{y+1})$ cannot exist for $z > y+1$.
\end{proof}

\begin{theorem}\label{theorem-exact2-alg}
Algorithm ExactNF2 finds a matching $M$ with $w(M)\geq c$ and $|M|$ maximized or concludes that there is no matching $M$ with $w(M)\geq c$ in $H$ in time $O(nm + n_{min}\cdot t(N))$, where $t(N)$ is the time for computing $SSSP(s)$ in a residual network of $\hat{N}$, $n=|V(N)|$ and $m=|E(N)|$.
\end{theorem}

\begin{proof}
From Property~\ref{property-reduced-cost-optimality}, the flow $f_y$ found by Step 3 at each iteration is a minimum cost flow of value $y$ in $\hat{N}$; and by Property~\ref{property-new-weight} (b), $f_y$ is a min-cost flow of value $y$ in $N$.
If Algorithm ExactNF2 terminates due to $f_y$ is maximum, then by Theorem~\ref{theorem-exact-alg}, either $f_y$ is an optimal solution, or there is no solution for $H$.

Suppose Algorithm ExactNF2 terminates due to $c(f_{y+1}) > c(f_y)$ and $c(f_{y+1}) > -c$.
By Lemma~\ref{lemma-concave-flow-cost}, any flow $f_z$ found after $f_{y+1}$ has $c(f_z) \geq c(f_{y+1}) > -c$.
If $c(f_{y}) \leq -c$, then $f_{y}$ implies an an optimal solution; otherwise, there is no solution for $H$.

Johnson's shortest path algorithm runs in $O(nm)$, due to the Bellman-Ford algorithm having a running time of $O(nm)$.
There are at most $n_{min}$ iterations after re-weighting the edges using Johnson's algorithm.
There are at most $n_{min}$ iterations.
In each iteration, it takes $t(N) \geq O(m+n\log n)$ time to compute $SSSP(s)$ and $O(m+n)$ time to update node potentials, reduced costs and residual network.
Therefore, Algorithm ExactNF2 runs in time $O(nm + n_{min}\cdot t(N))$.
\end{proof}

In the worst case, ExactNF2 has $n_{min}$ iterations (after re-weighting the edges). The stopping condition ($c(f_{y+1})>c(f_y)$ and $c(f_{y+1})>-c$) can reduce the number of iterations in practice.
Dijkstra's algorithm can be used to compute $SSSP(s)$ in each iteration since the edge cost in $\hat{N}_{f_y}(\pi)$ is non-negative.
As described in~\cite{Ahuja-NF93}, one can compute an $s-t$ path instead of computing $SSSP(s)$ to improve the computational time. 
After an $s-t$ path is found using Dijkstra's algorithm (with early termination), the update to each node potential $\pi(u)$ is based on whether vertex $\pi(u)$ has been permanently ($u$ is visited and explored) or temporarily labeled by Dijkstra's algorithm.
We implement this version in our experiment. 
Note, however, that Algorithm ExactNF2 still runs in time $O(nm + n_{min}\cdot t(N))$ asymptotically.
Due to the structure of the graph $N$, the length of any $s-t$ path in $N$ is exactly 3.
We can terminate the Bellman-Ford algorithm after 3 iterations (instead of $n-1$ iterations), resulting in time $O(m)$ for Johnson's shortest path algorithm. Hence, we have a reduced running time of $O(m + n_{min}\cdot t(N))$.

\subsection{Approximation algorithm}
Next, we present a simple $\frac{1}{2}$-approximation algorithm (referred to as \textbf{Greedy}) for an arbitrary profit target $c$. This algorithm is useful, in terms of actual computational time, if the instance $H$ contains many negative edges.
\begin{enumerate}\setlength\itemsep{0em}
\item Compute a maximum weight matching $M'$ in $H$.
\item Let $M = M'$.  For each iteration, select an edge $e''$ in $H^-=H\setminus H^+$ such that
\[
e'' = \text{argmax}_{e \in E(H^-) \setminus M\; \mid\; e \cap e' = \emptyset\; \forall e'\in M} w(e).
\] 
If $w(M) + w(e'') \geq c$, then add $e''$ to $M$. Repeat this until such an edge $e''$ does not exist (every edge of $H$ intersects with an edge of $M$) or $w(M) + w(e'') < c$.
\end{enumerate}
Algorithm Greedy has a running time of $O(t(H)+m\log m)$, where $t(H)$ is the time to find a maximum weight matching in $H$ and $m=|E(H)|$.

Next, we show Algorithm Greedy has a $\frac{1}{2}$-approximation ratio.
It is easy to see that Greedy always produces a feasible solution (matching).
Let $M'$ be the initial maximum weight matching computed in the first step of Greedy. Note that $M' \subseteq E(H^+)$.
Let $M^*$ be a matching in $H$ with $w(M^*) \geq c$ and $|M^*|$ maximized.

\begin{property} \label{property-greedy-initial-matching}
Every edge $e \in E(H^+) \setminus M'$ is incident to at least one edge of $M'$, implying any edge $e$ of $H$ not incident to $M'$ must have weight $w(e) < 0$.
\end{property}

Let $M_1=M\setminus M'$ be the set of edges added to $M'$ during the second step of Greedy.
Let $M^*_1=\{e \in M^* \mid e \text{ is not incident to any edge of } M'\}$.
From Property~\ref{property-greedy-initial-matching}, every $e\in M_1 \cup M^*_1$ has weight $w(e)<0$, namely, $(M_1 \cup M^*_1) \subseteq E(H^-)$.
From Property~\ref{property-greedy-initial-matching} and each edge of $M'$ is incident to at most two edges of $M^*\setminus M^*_1$, $|M'| \geq |M^*\setminus M^*_1|/2$ since $M'$ is a maximal matching in $H^+$.
Since $w(M')$ is maximum among all matchings in $H$, $w(M') \geq w(M^*\setminus M^*_1)$.

\begin{theorem}
Let $M$ be the matching found by the Greedy algorithm and $M^*$ be a matching in $H$ with $w(M^*) \geq c$ and $|M^*|$ maximized. Then, $\frac{|M|}{|M^*|} \geq \frac{1}{2}$, implying Greedy is $\frac{1}{2}$-approximate to RPC1.
\label{theorem-greedy-app}
\end{theorem}

\begin{proof}
First, divide $F(V,E) = M_1 \Delta M^*_1$ into two collections of components:
$\mathcal{F}^*_1=\{C\in \mathcal{F}\mid E(C)\subseteq M^*_1\}$ and $\mathcal{F}_0=\mathcal{F}\setminus\mathcal{F}^*$.
Notice that $|E(\mathcal{F}^*_1)|=|\mathcal{F}^*_1|$ and $(E(F)\cap M_1) \subseteq E(\mathcal{F}_0)$.
We prove that $E(F)\cap M_1$ can be divided into two subsets $E_0$ and $E_1$ such that $|E_0|\geq |E(\mathcal{F}_0) \cap M^*_1|/2$ and $|E_1|\geq |E(\mathcal{F}^*_1)|$. 
Initially, $E_0=\emptyset$ and $E_1=\emptyset$.
For each component $C \in \mathcal{F}_0$, let $e_1,e_2,\ldots,e_b$ be the set of edges in $E(C) \cap M_1$ selected by Greedy, where $e_i$ is the $i^{th}$ edge added to $E(C)\cap M_1$ and $b=|E(C)\cap M_1|$.
We divide $M_1 \cap E(\mathcal{F}_0)$ into $E_0$ and $E_1$ as follows:
for each $C\in \mathcal{F}_0$ and $i=1,\ldots,b$, if there is an edge in $E(C)$ incident to $e_i$ then add $e_i$ to $E_0$ and remove every edge in $E(C)$ incident to $e_i$ from $E(C)$ (each removed edge is in $E(C)\cap M^*_1$); otherwise, add $e_i$ to $E_1$.
Since each edge $e_i\in E_0$ is incident to at most two edges of $M^*_1$, $|E_0|\geq |E(\mathcal{F}_0)\cap M^*_1|/2$.
By the Greedy algorithm, $w(E_0)\geq w(E(\mathcal{F}_0) \cap M^*_1)$.
From this and $w(M')\geq w(M^*\setminus M^*_1)$, if there is any edge $e^*$ in $E(\mathcal{F}^*_1)$ then there is a distinct edge $e$ in $E_1$ with $w(e)\geq w(e^*)$, implying $|E_1|\geq |E(\mathcal{F}^*_1)|$.
Therefore,
{\small
\[
\frac{|M_1|}{|M^*_1|} = \frac{|E_0| + |E_1| + |M_1\cap M^*|}{|E(\mathcal{F}_0) \cap M^*_1| + |E(\mathcal{F}^*_1)| + |M_1\cap M^*|} \geq \frac{|E(\mathcal{F}_0) \cap M^*_1|/2 + |E(\mathcal{F}^*_1)|+ |M_1\cap M^*|}{|E(\mathcal{F}_0) \cap M^*_1| + |E(\mathcal{F}^*_1)| + |M_1\cap M^*|} \geq \frac{1}{2}.
\]
}%
From this and $|M'| \geq |M^*\setminus M^*_1|/2$,
\begin{align*}
\frac{|M|}{|M^*|} = \frac{|M'|+|M_1|}{|M^*\setminus M^*_1| + |M^*_1|} \geq \frac{|M^*\setminus M^*_1|/2+|M^*_1|/2}{|M^*\setminus M^*_1| + |M^*_1|} \geq \frac{1}{2}. \tag*{\qedhere}
\end{align*}%
\end{proof}

\section{RPC+ variant} \label{sec-app-nonnegative}
Due to the non-negative profit constraint~\eqref{constraint-nonnegative-edge}, only edges in $H^+$ can be selected to solve the RPC\texttt{+} problem (formulation \eqref{formulation-max-passenger}-\eqref{constraint-nonnegative-edge})
Inherently, the profit target must be non-negative for the RPC\texttt{+} problem.
In this case, a matching $M$ with $w(M) \geq c$ and $|M|$ maximized may not be an optimal solution to the RPC\texttt{+} problem if $\lambda \geq 2$.
For instance, a matching $M_1=\{e_1,e_2,e_3\}$ with three edges may contain only three passenger vertices of $V(H) \cap R$, whereas a matching $M_2=\{e_4\}$ with one edge can contain four passenger vertices (assuming $\lambda\geq 4$).
We need to find a matching $M$ in $H$ such that the number of passenger vertices $V(H) \cap R$ contained in $M$ is maximized and $w(M) \geq c$.

\subsection{Algorithm}
We propose a local search algorithm for $\lambda \geq 2$, called \textbf{LS2}.
For an edge $e=\{\eta_i\} \cup R_i \in E(H)$, let $D(e)=\{\eta_i\}$ (the driver of edge $e$) and $R(e)=R_i$ (the passengers of $e$).
For a subset $E' \subseteq E(H)$, let $D(E') = \cup_{e \in E'} D(e)$ and $R(E') = \cup_{e \in E'} R(e)$.
For an edge $e \in E(H)$, let $N(e)$ be the set of edges incident to $e$, called the \emph{neighborhood} of $e$, and $N^+(e) = N(e) \cap E(H^+)$.
By constraint~\eqref{constraint-nonnegative-edge}, we only need to consider the subgraph $H^+$.
Although the profit target $c$ is an input parameter of an RPC\texttt{+} instance, one needs to determine that $c$ should be at most $c^*$ (the weight of a maximum weight matching in $H$) so that it admits a feasible solution.
However, finding $c^*$ is NP-hard as mentioned in the preliminaries.
Note that as $c$ gets closer to $c^*$, the chance of having a lower objective value is higher.
We suggest a way to set $c$ to be a reasonable target value.
We use a heuristic to compute a weight $\tilde{c}$ to approximate $c^*$ and set $c\leq \tilde{c}$ as a profit target.
In fact, our experiment shows that the total profit of solutions with respect to $c$ is not too far way from the total profit of solutions with respect to $c^*$ in practice.

There are two steps in Algorithm LS2.
In the first step, LS2 uses the \textit{simple greedy} in~\cite{Berman-SWAT00, Chandra-JoA01} to find an initial weighted set packing (hypergraph matching) to get $\tilde{c}$.
In the second step of LS2, a local search is used to improve the solution computed in the first step.
The first step produces a solution with a $\frac{1}{2\lambda}$-approximation ratio, and the second step gives a solution with a $\frac{2}{3\lambda}$-approximation ratio when a specific condition on the profit target is met.
Algorithm LS2 is given in the following, starting with $M'=\emptyset$.
\begin{enumerate}
\item In each iteration, select an edge $e'' \in E(H^+)$ that does not intersect with any edge of $M'$ and has maximum weight. That is, find an edge $e''$ in $E(H^+)$ such that 
\[
e'' = \text{argmax}_{e \in E(H^+)\setminus M'\; \mid\; e \cap e' = \emptyset\; \forall e'\in M'} w(e),
\] and add $e''$ to $M'$.
Repeat this until every edge of $E(H^+) \setminus M'$ intersects with $M'$.
Determine $c$ by setting $c \leq \tilde{c} = w(M')$.

\item Let $M = M'$ be the matching obtained after Step 1.
Let $A =\{e \in M \mid |R(e)|=1\}$ and assume $A=\{a_1,\ldots,a_q\}$ with $w(a_i)\leq w(a_j)$ for $1\leq i<j\leq q$.
An \emph{improvement} $\delta_e$ of an edge $e \in M$ is a subset of edges in $N^+(e)$ such that
\begin{itemize}\setlength\itemsep{0em}
\item $|\delta_e| \leq 2$, all edges of $(M \cup \delta_e)\setminus \{e\}$ are pairwise vertex-disjoint, $|R(\delta_e)| > |R(e)|$ and $w(M) + w(\delta_e) - w(e) \geq c$.
\end{itemize}
An improvement $\delta_e$ is \emph{maximum} if $|R(\delta_e) \setminus R(M)|$ is maximum among all improvements of $e$.
\begin{enumerate}\setlength\itemsep{0em}
\item If $\lambda = 2$, execute the following for-loop for each $a_i \in A$.
    \begin{itemize}
          \item For $i=1$ to $q$ do, if there is an improvement $\delta_{a_i}$ of $a_i$ such that $|R(\delta_{a_i})| = 4$, then perform an \emph{augmentation} as $M = (M \cup \delta_{a_i}) \setminus \{a_i\}$.
    \end{itemize}
\item Else if $\lambda \geq 3$, execute the following for-loop for each $a_i \in A$.
    \begin{itemize}
          \item For $i=1$ to $q$ do, if there is an improvement of $a_i$, then find a maximum improvement $\delta_{a_i}$ and perform an \emph{augmentation} as $M = (M \cup \delta_{a_i}) \setminus \{a_i\}$.
    \end{itemize}
\end{enumerate}
Output $M$.
\end{enumerate}

\subsection{Analysis of LS2}
Let $M'$ be the solution (a matching) found after Step 1 of the LS2 algorithm.
Let $0 \leq c \leq w(M')$ and $M^*$ be a matching in $H^+$ such that $|R(M^*)|$ is maximized and $w(M^*) \geq c$, representing an optimal solution to the RPC\texttt{+} problem.
We first show that $M'$ is already $\frac{1}{2\lambda}$-approximate for any $\lambda \geq 1$.

\begin{property}
Every edge $e \in E(H^+) \setminus M'$ is incident to at least one edge of $M'$.
\label{property-incident-edges}
\end{property}

For any matching $M''$ found during the execution of algorithm LS2 and an edge $e \in M''$, let $M^*(e) = \{e^* \in M^* \mid e^* \text{ is incident to } e\}$.
For two incident edges $e \in M'$ and $e^* \in M^*$, we say they are incident/intersected \textit{by trip} if $R(e) \cap R(e^*) \neq \emptyset$, \textit{by driver} if $D(e) = D(e^*)$, or \textit{by both} if $R(e) \cap R(e^*) \neq \emptyset$ and $D(e) = D(e^*)$.

\begin{theorem}
Let $M'$ be the matching found by Step 1 of the LS2 algorithm, $0 \leq c \leq w(M')$ and $M^*$ be a matching in $H^+$ such that $|R(M^*)|$ is maximized and $w(M^*) \geq c$.
Then, $\frac{|R(M')|}{|R(M^*)|} \geq \frac{1}{2\lambda}$ for $\lambda \geq 1$.
\label{lemma-2lambda-app-nonnegative}
\end{theorem}

\begin{proof}
For every $e \in M'$, $e$ is incident to at most one edge of $M^*$ by driver and $|R(e)|$ edges of $M^*$ by trip.
From this, 
\[
\frac{|R(e)|}{|R(M^*(e))|} \geq \frac{|R(e)|}{(|R(e)|+1)\lambda} \geq \frac{1}{2\lambda}.
\]
From Property~\ref{property-incident-edges}, every edge $e^* \in M^*$ must be incident to some edge $e \in M'$ (or $e^* \in M'$).
Therefore, $\frac{|R(M')|}{|R(M^*)|} \geq \frac{1}{2\lambda}$.
\end{proof}

As can be seen in the analysis of Theorem~\ref{lemma-2lambda-app-nonnegative}, the approximation ratio is dominated by $\frac{|R(e)|}{|R(M^*(e))|} \geq \frac{|R(e)|}{(|R(e)|+1)\lambda}$ for every $e \in M'$.
We can generalize it as the following corollary.
\begin{corollary}
If $|R(e)| \geq b$ for every edge $e \in M$ and a constant $b \geq 1$, then $\frac{|R(M)|}{|R(M^*)|} \geq \frac{b}{(b+1)\lambda}$.
\label{corollary-general-app-nonnegative}
\end{corollary}%

If the number of edges in $M'$ containing only one passenger is small, then a better approximation ratio can be obtained, which is the main purpose of Step 2 of LS2.
Recall that $A =\{e \in M' \mid |R(e)|=1\}$.
When an edge $e$ of $A$ is replaced by an improvement $\delta_e$, $|R(\delta_e)| - |R(e)| \geq 1$, increasing $|R(M')|$ by at least one.
However, an improvement $\delta_e$ can decrease the total weight $w(M')$.
If the profit target $c$ is smaller than $w(M')$ by some fraction of $w(A)$ (as stated in Assumption~\ref{assumption-profit-target-range}), then we can get a $\frac{2}{3\lambda}$-approximation, and we prove this in the remainder of this section.

\begin{assumption}
The profit target $c$ is within $0 \leq c \leq w(M' \setminus A) + 2w(A) / (\lambda+1)$.
\label{assumption-profit-target-range}
\end{assumption}

Let $F(V,E) = M' \Delta M^*$ be the resulting graph of the symmetric difference of $M'$ and $M^*$ for the rest of the analysis.
Let $\mathcal{F}$ be the set of connected components in $F(V,E)$.
Let $\mathcal{C}(3) =\{C \in \mathcal{F} \mid |E(C)| = 3 \text{ and } |R(M'\cap E(C))|=1 \text{ and } |R(e^*)| > 1 \text{ for each } e^*\in M^*\cap E(C)\}$.
Let $Q = \cup_{C \in \mathcal{C}(3)} M'\cap E(C)$ and $Q^*=\cup_{C\in \mathcal{C}(3)} M^*\cap E(C)$.
By the definition of $\mathcal{C}(3) $, every $C \in \mathcal{C}(3)$ contains exactly one edge of $M'$ (due to $|R(M'\cap E(C))|=1$) and two edges $e^*_1$ and $e^*_2$ of $M^*$.
Then, there exists an improvement $\delta_e$ with $\delta_e=\{e^*_1\}$, $\delta_e=\{e^*_2\}$, or $\delta_e=\{e^*_1, e^*_2\}$ for every $e \in Q$.
Such an improvement $\delta_e$ is independent of each $C\in \mathcal{C}(3)$.
For an edge $e\in Q$ and an improvement $\delta_e$ with $|\delta_e|=2$ w.r.t. $M'$, $|R(\delta_e)| \geq |R(e^*_1\cup e^*_2)|-|R(e)| \geq 3$ because $|R(e^*_1)| \geq 2$ and $|R(e^*_2)| \geq 2$ by definition.
Similarly for an improvement $\delta_e$ with $|\delta_e|=1$, $|R(\delta_e)| -|R(e)| \geq 1$.
Since each edge of $Q$ is incident to two edges of $Q^*$ and any edge $e^*$ of $Q^*$ is only incident to one edge of $Q$, $|Q^*| = 2|Q|$.
Hence, we have the following property.

\begin{property} \label{property-Q}
Let $\mathcal{C}(3) $, $Q$ and $Q^*$ be defined as above.
\begin{enumerate}
\item[(1)] There exists an improvement $\delta_e$ w.r.t. $M'$ for every $e \in Q$ such that edges of $\delta_e$ are not incident to edges of $\delta_{e'}$ for every pair $e,e' \in Q$. 
\item[(2)] $|R(Q)| \geq \frac{|R(Q^*)|}{2\lambda}$.
\end{enumerate}
\end{property}

\begin{lemma}
$|R(M' \setminus Q)| \geq \frac{2}{3\lambda} |R(M^* \setminus Q^*)|$.
\label{lemma-property-ratio}
\end{lemma}
\begin{proof}
Let $\mathcal{F}_1 = \{C \in \mathcal{F} \setminus \mathcal{C}(3) \mid |R(M'\cap E(C))|=1\}$ and $\mathcal{F}_2 = \mathcal{F} \setminus (\mathcal{F}_1 \cup \mathcal{C}(3))$.
Let $M'_1 = \cup_{C \in \mathcal{F}_1} M'\cap E(C)$, $M^*_1=\cup_{C\in \mathcal{F}_1} M^*\cap E(C)$, $M'_2 = \cup_{C \in \mathcal{F}_2} M'\cap E(C)$ and $M^*_2=\cup_{C\in \mathcal{F}_2} M^*\cap E(C)$.
We first consider (1) $\mathcal{F}_1$, and then (2) $\mathcal{F}_2$.

(1) For each $C \in \mathcal{F}_1$, $C$ has exactly one edge $e$ of $M'$ and at most two edges of $M^*$.
If $C$ contains two edges $e^*_1$ and $e^*_2$ of $M^*$, then one of $e^*_1$ and $e^*_2$ contains only one passenger by the definition of $\mathcal{C}(3)$.
This implies that 
\[
\frac{|R(e)|}{|R(e^*_1)\cup R(e^*_2)|} \geq \frac{1}{\lambda+1} \geq \frac{2}{3\lambda} \hspace*{4mm} \text{ for } \lambda \geq 2.
\]
If $C$ contains only one edge $e^*$ of $M^*$, $|R(e)|/|R(e^*)| \geq 1/\lambda$.
From these, $|R(M'_1)|\geq \frac{2}{3\lambda} |R(M^*_1)|$.

(2) Let $\mathcal{F}'_2=\{C\in \mathcal{F}_2 \mid |M'\cap E(C)|=1\}$.
For each $C \in \mathcal{F}'_2$ and an edge $e \in M'\cap E(C)$, $|R(e)| \geq 2$ by the definition of $\mathcal{F}_1$.
Then, 
\begin{align}
\frac{|R(e)|}{|R(M^*(e))|} \geq \frac{|R(e)|}{(|R(e)| +1)\lambda} \geq \frac{2}{3\lambda}. \label{eq-F2-1}
\end{align}
Let $\mathcal{F}''_2=\{C\in \mathcal{F}_2 \mid |M'\cap E(C)|\geq2\}$.
For any $C\in \mathcal{F}''_2$ and every two edges $e_1$ and $e_2$ in $M' \cap E(C)$ that are connected by an edge of $M^*\cap E(C)$, there can be at most $|R(e_1)\cup R(e_2)| +1$ edges of $M^*$ incident to $e_1$ and $e_2$.
From this, 
\begin{align}
\frac{|R(e_1\cup e_2)|}{|R(M^*(e_1)\cup M^*(e_2))|} \geq \frac{|R(e_1)\cup R(e_2)|}{(|R(e_1)\cup R(e_2)| +1)\lambda} \geq \frac{2}{3\lambda}. \label{eq-F2-2}
\end{align}
Eq~\eqref{eq-F2-1} and Eq~\eqref{eq-F2-2} imply that $\frac{|R(M'_2)|}{|R(M^*_2)|} \geq \frac{2}{3\lambda}$.
Since $M'\setminus Q = M'_1 \cup M'_2 \cup (M'\cap M^*)$ and $M^* \setminus Q^* = M^*_1 \cup M^*_2 \cup (M'\cap M^*)$, $|R(M' \setminus Q)| \geq \frac{2}{3\lambda} |R(M^* \setminus Q^*)|$ from the above.
\end{proof}

Let $B = \{a_1, a_2, \ldots, a_m\}$ be any subset of $A$ such that $w(a_1) \leq w(a_i) \leq w(a_j) \leq w(a_m)$ for $1\leq i<j \leq m$, $m\geq 2$ and $\lambda \geq 2$.
For two integers $1\leq x \leq y\leq m$, let $B(x,y) = \{a_x, \ldots, a_y\}$.
Since $B \subseteq A$, we know that
\[
w(A) - w(B) \geq z\cdot(w(A) - w(B)) = z\cdot w(A) - z\cdot w(B)
\]
for any constant $0\leq z \leq 1$.
Hence, we have the following property.

\begin{property}
$\frac{2}{\lambda+1}w(B) - w(B) \geq \frac{2}{\lambda+1}w(A) - w(A)$ for $\lambda \geq 1$.
\label{property-A-B}
\end{property}

\begin{lemma}
Let $B'=B(1, \floor{\frac{(\lambda-1)m}{\lambda+1}})$ for $\lambda \geq 2$.
For matching $M'' = (M' \setminus B') \cup E_{B'}$, where $E_{B'} \subseteq (E(H^+)\setminus B')$ is a (an empty) set of edges incident to $B'$ s.t. $M''$ remains as a matching, $w(M'') \geq c$.
\label{lemma-minimum-A1-size}
\end{lemma}

\begin{proof}
Each edge $e \in B$ that is replaced by the algorithm can reduce $w(M')$ by at most $w(e)$ since the improvement $\delta_e$ of $e$ has weight $w(\delta_e) \geq 0$.
Recall that the elements of $B$ are sorted in the increasing order of their weights.
For $B'=B(1,\floor{\frac{(\lambda-1)m}{\lambda+1}})$, $B'$ contains at most $\frac{\lambda-1}{\lambda+1}$ smallest elements in $B$.
Hence,
\[
w(B') \leq \frac{\lambda-1}{\lambda+1}w(B)=(1-\frac{2}{\lambda+1})w(B).
\]
Let $M''=(M' \setminus B') \cup E_{B'}$.
From $w(E_{B'})\geq 0$, Property~\ref{property-A-B} and Assumption~\ref{assumption-profit-target-range},  
\begin{align*}
w(M'')\geq w(M')-w(B') &\geq w(M')-(1-\frac{2}{\lambda+1})w(B)\\
&= w(M'\setminus B)+\frac{2}{\lambda+1}w(B) \\
&\geq w(M'\setminus A)+\frac{2}{\lambda+1}w(A) \geq c. \tag*{\qedhere}
\end{align*}
\end{proof}

During Step 2 of Algorithm LS2, an edge $e^* \in Q^*$ remains \emph{unblocked} if $e^*$ is not incident to any edge of improvement $\delta_{e}$ for any $e \in A$; and, after the augmentation of any improvement $\delta_{e}$ for some $e \in A$, an unblocked edge $e^* \in Q^*$ becomes \emph{blocked} if $e^*$ is incident to any edge of $\delta_{e}$, that is, $D(e^*) \cap D(\delta_{e}) \neq \emptyset$ or $R(e^*) \cap R(\delta_{e}) \neq \emptyset$.

\begin{lemma}
Let $M$ be the final matching found by the LS2 algorithm.
Let $M^*$ be a matching in $H^+$ such that $|R(M^*)|$ is maximized and $w(M^*) \geq c$, representing an optimal solution to the RPC\texttt{+} problem. 
For $\lambda = 2$, $\frac{|R(M)|}{|R(M^*)|} \geq \frac{2}{3\lambda}$.
\label{lemma-lambda2-app}
\end{lemma}

\begin{proof}
We show that at least $|Q|/3$ additional edges are added to $M'$, that is, $|R(M)| - |R(M')| \geq |Q|/3$.
For $\lambda = 2$, removing $|B(1, \floor{\frac{(\lambda-1)m}{\lambda+1}})| \geq |B(1, \floor{\frac{m}{3}})|$ edges of $A$ from $M'$ results in a matching $M''$ with $w(M'') \geq c$ by Lemma~\ref{lemma-minimum-A1-size}.
Letting $m = |A|-\floor{2|A|/3} \geq \ceil{|A|/3} \geq \ceil{|Q|/3}$ implies that $\floor{\ceil{|A|/3}/3} \geq \ceil{|Q|/9}$ improvements (in augmentations) can be performed on $M'$ such that the resulting matching $M''$ has weight $w(M'') \geq c$.
By the definition of $\mathcal{C}(3)$, each improvement $\delta_e$ of $e \in Q$ contains the two edges $e^*_1, e^*_2 \in Q^*$ incident to $e$. Let $\delta_Q=\{\delta_e \mid e\in Q\}$ be the set of such improvements.
An improvement $\delta_e$ for $e \in A \setminus Q$ adds 3 additional edges to $M'$ and can be incident to at most 4 improvements of $\delta_Q$ since $e$ is not incident to any edge of $Q^*$.
An improvement $\delta_e$ for $e \in Q$ adds at least 3 additional edges to $M'$ and can be incident to at most 5 improvements of $\delta_Q$ (four different improvements plus the one already incident to $e$).
This implies that, in the worst case, after $\ceil{|Q|/9}$ improvements with $|Q|/3$ additional edges, there are still some improvements of $\delta_Q$ not incident to these $\ceil{|Q|/9}$ improvements.
Hence, at least $|Q|/3$ additional edges are added to $M'$.

Recall from Lemma~\ref{lemma-property-ratio} and Property~\ref{property-Q} (2) that $|R(M' \setminus Q)| \geq \frac{2}{3\lambda}|R(M^* \setminus Q^*)|$ and $|R(Q)| \geq |R(Q^*)|/2\lambda$.
We have
\begin{align} \label{eq-lambda3-app-1}
|R(M)| \geq |R(M' \setminus Q)| + |R(Q)| + \frac{|Q|}{3} &= |R(M' \setminus Q)| + \frac{4|R(Q)|}{3} \\
&\geq \frac{2}{3\lambda}|R(M^* \setminus Q^*)| + \frac{2}{3\lambda}|R(Q^*)| \nonumber
\end{align}
From Eq~\eqref{eq-lambda3-app-1},
\begin{align} \label{eq-lambda3-app-2}
\frac{|R(M)|}{|R(M^*)|} \geq \frac{\frac{2}{3\lambda}|R(M^* \setminus Q^*)| + \frac{2}{3\lambda}|R(Q^*)|}{|R(M^* \setminus Q^*)| + |R(Q^*)|}
=\frac{2}{3\lambda}.
\end{align}
Therefore, the lemma holds.
\end{proof}

Lemma~\ref{lemma-lambda3-app} is proved in a similar way as the proof of Lemma~\ref{lemma-lambda2-app}.
\begin{lemma}
Let $M$ be the final matching found by the LS2 algorithm.
Let $M^*$ be a matching in $H^+$ such that $|R(M^*)|$ is maximized and $w(M^*) \geq c$, representing an optimal solution to the RPC\texttt{+} problem.
For $\lambda \geq 3$, $\frac{|R(M)|}{|R(M^*)|} \geq \frac{2}{3\lambda}$.
\label{lemma-lambda3-app}
\end{lemma}
\begin{proof}
We show that at least $|Q|/3$ additional edges are added to $M'$, that is, $|R(M)| - |R(M')| \geq |Q|/3$.
For $\lambda \geq 3$, removing $|B(1, \floor{\frac{(\lambda-1)m}{\lambda+1}})| \geq |B(1, \floor{\frac{m}{2}})|$ edges of $A$ from $M'$ results in a matching $M''$ with $w(M'') \geq c$ by Lemma~\ref{lemma-minimum-A1-size}.
Letting $m = |A|-\floor{|A|/3} \geq \ceil{2|A|/3} \geq \ceil{2|Q|/3}$ implies that $\floor{\ceil{2|Q|/3}/2} \geq \ceil{|Q|/3}$ improvements (in augmentations) can be performed on $M'$ such that the resulting matching $M''$ has weight $w(M'') \geq c$.
By Property~\ref{property-Q} (1), there are $2|Q|$ improvements w.r.t. $M'$, each of which is an edge in $Q^*$.
An improvement $\delta_e$ for $e \in A \setminus Q$ adds $|R(\delta_e)|-1 \geq 1$ additional edges to $M'$ and can be incident to at most $|R(\delta_e)|$ unblocked edges of $Q^*$ since $e$ is not incident to any edge of $Q^*$; and these unblocked edges become blocked and are no longer improvements after augmenting $\delta_e$.
Similarly, an improvement $\delta_e$ for $e \in Q$ with $|\delta_e| =2$ adds at least $|R(\delta_e)|-1 \geq 3$, due to maximum improvement, additional edges to $M'$ and can be incident to at most $|R(\delta_e)|+2$ unblocked edges of $Q^*$.
It is possible that after some augmentations, an improvement $\delta_e$ for $e \in Q$ contains only one edge due to blocked edges of $Q^*$; and in this case, each such improvement $\delta_e$ adds at least $|R(\delta_e)|-1 \geq 1$ additional edges to $M'$ and can be incident to at most $|R(\delta_e)|+1$ unblocked edges of $Q^*$.
These imply that, in the worst case, even after $\ceil{|Q|/3}$ improvements with $|Q|/3$ additional edges, not all edges of $Q^*$ are incident to some edges of these $\ceil{|Q|/3}$ improvements, that is, there are edges of $Q^*$ remain unblocked.
Hence, at least $|Q|/3$ additional edges are added to $M'$.
Therefore, by Eq~\eqref{eq-lambda3-app-1} and Eq~\eqref{eq-lambda3-app-2}, the lemma holds.
\end{proof}

From Theorem~\ref{lemma-2lambda-app-nonnegative}, Lemma~\ref{lemma-lambda2-app} and Lemma~\ref{lemma-lambda3-app}, we have Theorem~\ref{theorem-ls2-app}.
\begin{theorem}
Let $M'$ be the matching found by Step 1 of the LS2 algorithm and $M$ be the final matching found by the LS2 algorithm.
Let $A = \{e \in M' \mid |R(e)|=1\}$.
Let $0 \leq c \leq w(M')$ and $M^*$ be a matching in $H^+$ such that $|R(M^*)|$ is maximized and $w(M^*) \geq c$.
$\frac{|R(M')|}{|R(M^*)|} \geq \frac{1}{2\lambda}$ for $\lambda \geq 1$, and if $c \leq w(M' \setminus A) + \frac{2w(A)}{\lambda+1}$ for $\lambda \geq 2$, then $\frac{|R(M)|}{|R(M^*)|} \geq \frac{2}{3\lambda}$.
\label{theorem-ls2-app}
\end{theorem}

\section{Numerical experiments} \label{sec-experiment}
We conduct an extensive empirical study to evaluate our model and algorithms for RPC1 and RPC\texttt{+}.
To the best of our knowledge, there is no practical test dataset publicly available for the RPC problem.
To clear this hurdle, we create a simulation dataset by incorporating a real-world ridesharing dataset from Chicago City with the driver profit model of Uber.
In Section~\ref{sec-experiment-overview}, we introduce the simulation setup and describe the ridesharing dataset of Chicago City. We describe how to apply the profit/revenue model of Uber to assign a profit/revenue value to each feasible match in Section~\ref{sec-experiment-profit}. In Section~\ref{sec-experiment-instance}, we describe the test instances in detail.
Experimental results are reported in Section~\ref{sec-experiment-results}.

\subsection{Simulation and dataset overview} \label{sec-experiment-overview}
The simulated centralized system receives a batch of driver offer trips $D$ and passenger request trips $R$ in a fixed time interval, where the origin $o_i$ and destination $d_i$ are within Chicago City for every driver $\eta_i\in D$ and passenger $r_i\in R$.
The roadmap data of Chicago city is retrieved from OpenStreetMap (BBBike.org)\footnote{Planet OSM. \url{https://planet.osm.org}.  BBBike. \url{https://download.bbbike.org/osm/}}.
We used the GraphHopper\footnote{GraphHopper 6.0. \url{https://www.graphhopper.com}} library to construct the logical graph data structure of the roadmap, which contains 290048 vertices and 414124 edges.

The ridesharing dataset (denoted by \textit{TNP}) we use is publicly available on Chicago Data Portal (CDP), maintained by Chicago Transit Authority (CTA)\footnote{CDP. \url{https://data.cityofchicago.org}. CTA. \url{https://www.transitchicago.com}}.
The TNP dataset contains completed trips records, reported by Transportation Network Providers (which are rideshare companies) to CTA.
The TNP dataset range is chosen from May 1st, 2022 to May 31st, 2022.
The Chicago city is divided into 77 official community areas (\textbf{area} for brevity and labeled as A1 to A77).
Each record in the TNP dataset describes a passenger trip served by a driver who provides the ride service.
A trip record in the TNP dataset contains:
\begin{itemize}\setlength\itemsep{0em}
\begin{small}
\item A pick-up time, a drop-off time, a pick-up area (Census Tract) and a drop-off area (Census Tract) for the passenger.
\item Duration and distance travelled of the trip.
\item A fare and an optional tip paid by the passenger, where the fare does not include the tip.
\end{small}
\end{itemize}
Note that the exact pick-up and drop-off locations are not provided from the dataset and times are rounded to the nearest 15 minutes.
We removed any trip record that is missing any of the essential information, ``short'' trips (less than 1.5 miles or 6 minutes) and ``long'' trips (greater than 35 miles or 70 minutes) from the dataset.
Our experiment focuses on weekdays only.
This results in 2453435 trip records from the TNP dataset.
We group one or more adjacent areas together to create 25 \emph{\textbf{regions}} to represent the Chicago City.
Areas of a region are grouped together based on the total number of trips with pick-up and drop-off areas in that region.
That is, areas with smaller number of trips are grouped together as shown in Figure~\ref{fig-regions}.
Since the trip records in TNP are aggregated in every 15 minutes, we partition a day from 6:00 to 23:59 into 72 time intervals (each has 15 minutes).
In each time interval, drivers and passengers are generated for each of the 25 regions.
\begin{figure}[hbtp]
\centering
\includegraphics[width=0.85\linewidth]{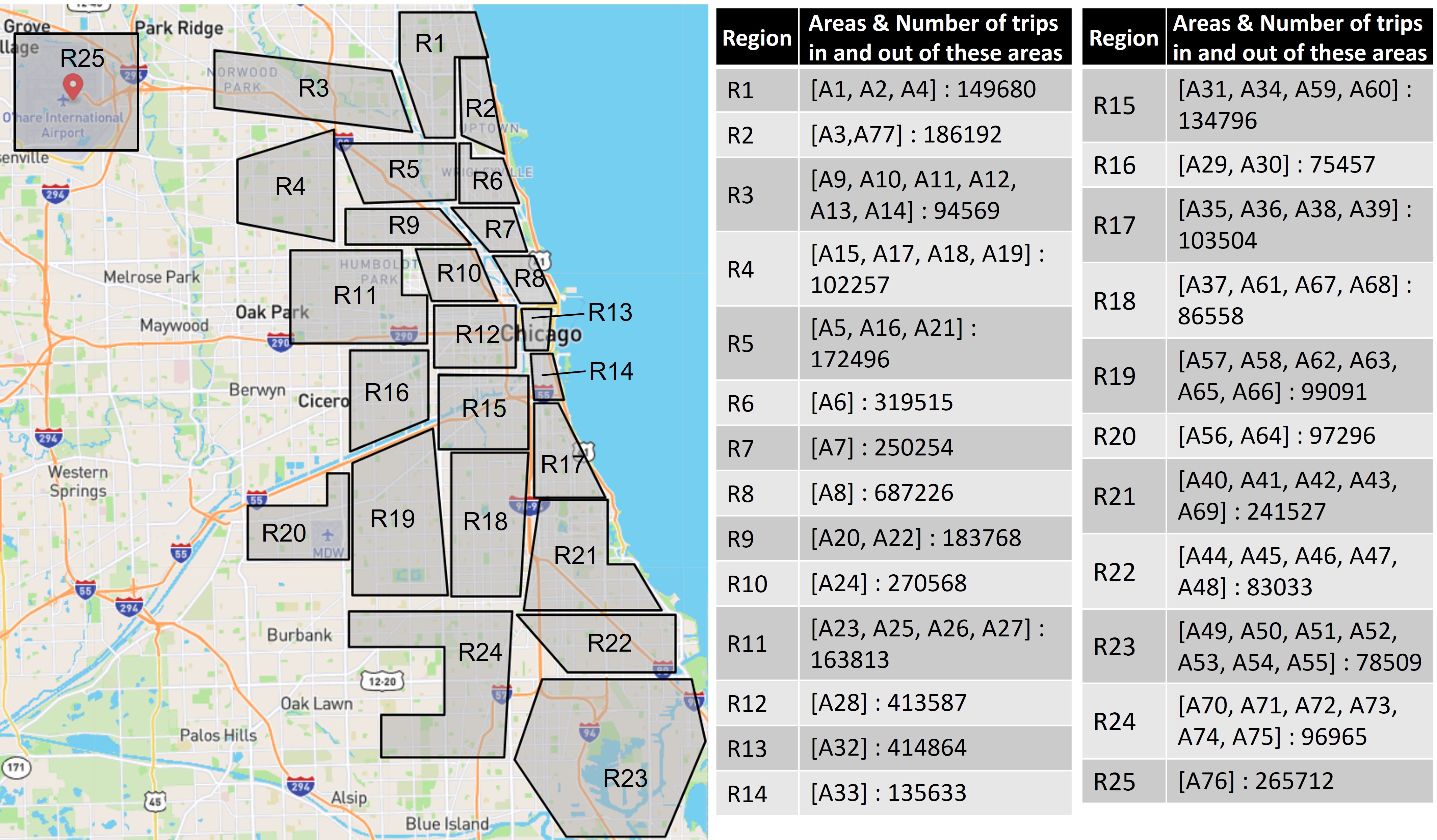}
\captionsetup{font=small}
\caption{The 77 community areas are grouped into 25 regions.}
\label{fig-regions}
\end{figure}

\subsection{Profit for feasible matches} \label{sec-experiment-profit}
First, we describe how we estimate the revenue $rev(\eta_i, R_i)$ of a feasible match $(\eta_i, R_i)$ for a driver $\eta_i$ using the TNP dataset.
Since the trip records in TNP are reported by rideshare companies in the US, we use the price scheme from Uber, based on its upfront cost estimator\footnote{Uber cost estimator. \url{https://www.uber.com/global/en/price-estimate}.} (Lyft has a similar scheme).
Recall that $\SFP(\eta_i, R_i)$ is the shortest feasible path $\eta_i$ needs to traverse to serve all of $R_i$.
Let $\zeta=|R_i|$ and $\SFP(\eta_i, R_i) = (l_0,l_1,\ldots,l_{2\zeta+1})$ such that for $1 \leq a \leq 2\zeta$, $l_a$ is a location (a vertex in the road network graph) representing either an origin $o_j$ or a destination $d_j$ of passenger $r_j$ in $R_i$, where $l_0 = o_i$ and $l_{2\zeta+1}=d_i$.
For $1\leq a<b\leq 2\zeta$, let $p(a,b)$ be the subpath of $\SFP(\eta_i, R_i)$ from $l_a$ to $l_b$.
For a passenger $r_j \in R_i$, let $p(j_1, j_q)$ be the subpath such that $l_{j_1} = o_j$ and $l_{j_q} = d_j$.
In other words, $p(j_1, j_q)$ is the path passenger $r_j$ needs to traverse.
Denoted by $t(l_a,l_b)$ is the estimated travel time (duration) for traversing $p(a,b)$.
We define a cost function $g(\eta_i,r_j)$ representing how much $r_j$ needs to pay.
The cost of a trip for a passenger from Uber's cost estimator includes: a base fare $f_1$, a per-minute cost multiplier $f_2$, a per-mile cost multiplier $f_3$ and a booking/service fee $f_4$.
Then the cost for passenger $r_j \in R_i$ is 
\[
g(\eta_i,r_j) = \gamma(r_j) \cdot [f_1 + f_2\cdot t(j_1,j_q) + f_3\cdot \dist(p(j_1,j_q))] + f_4,
\]
where $\gamma(r_j)\geq 1$ is a surge pricing factor. Surge pricing factor $\gamma(r_j)$ fluctuates based on passenger-demand and driver-supply, which depends on when $r_j$ is picked-up, $o_j$ and $d_j$.
Let $f(\eta_i,r_j) = g(\eta_i,r_j) - f_4$.
Uber takes all of the booking fee $f_4$ and takes a portion of $f(\eta_i,r_j)$, which is known as the ``take-rate'' $\theta(r_j, R_i)$, and it is usually $0.2\leq\theta(r_j, R_i)\leq 0.25$.
In addition, $r_j$ has the option to include a tip $\epsilon(r_j) > 0$ for driver $\eta_i$.
The estimated revenue $rev(\eta_i, R_i)$ for driver $\eta_i$ and $|R_i|=\{r_j\}$ is $ (1-\theta(r_j, R_i))\cdot \sum_{r_j \in R_i} f(\eta_i,r_j) + \epsilon(r_j)$.

If $|R_i| > 1$, the match $(\eta_i,R_i)$ may become shared trips.
The price scheme is similar, except a discounted rate $\omega(r_j,R_i)$ is applied to $f(\eta_i,r_j)$, where $0\leq \omega(r_j,R_i) \leq 1$ and $\omega(r_j,R_i)=1$ means no discount.
Let $dp(r_j, R_i)$ be the number of different passengers in $R_i\setminus\{r_j\}$ encountered by $r_j$ while traversing $p(j_1, j_q)$.
We set $\omega(r_j,R_i) = \max\{1.0 - 0.2dp(r_j, R_i), 0.2\}$ (a linear relation).
In addition, the cost for shared distance travelled (per-minute and per-mile) is split among the passengers in the car at the time.
For $1\leq a \leq 2\zeta-1$, let $np(l_a,l_{a+1})$ be the number of passengers in the car while travelling from $l_a$ to $l_{a+1}$.
The take-rate $\theta(r_j, R_i)$ for shared trips (e.g., UberPool) can be set lower to adjust for the earnings of the drivers:
take-rate $\theta(r_j, R_i)$ for a passenger $r_j\in R_i$ is selected uniformly at random from $[\max\{0.05, 0.2\omega(r_j,R_i)\},\max\{0.1,0.25\omega(r_j,R_i)\}]$.
The final estimated revenue $rev(\eta_i, R_i)$ is
\begin{small}
\begin{align*}
\sum_{r_j \in R_i} (1-\theta(r_j, R_i)) \cdot \omega(r_j,R_i) \cdot \gamma(r_j) \cdot
(f_1 + \sum_{j_1\leq a \leq j_{q-1}} \frac{f_2\cdot t(l_a,l_{a+1}) + f_3\cdot \dist(l_a,l_{a+1})}{np(l_a,l_{a+1})}) + \epsilon(r_j).
\end{align*}
\end{small}%

Next, we describe how we estimate the fee components, surge pricing factor, travel time, tip amount, and profit $w(\eta_i, R_i)$ in detail.
From Uber's cost estimator, we can determine that $f_1$, $f_2$ and $f_3$ are fixed regardless of the distance of the trip.
The booking fee increases as the estimated distance of the trip increases.
\begin{table}[htbp]
\small
\centering
   \begin{tabular}{| c | c | c | c |}
   	\hline
             base fare $f_1$  & per-minute $f_2$   & per-mile $f_3$    &  booking fee $f_4$            \\ \hline
             1.8           & 0.27         & 0.8         & $\min\{\max\{1,1+0.25(\text{miles}- 2)\},10\}$      \\ \hline
   \end{tabular}
\caption{The cost (in USD) for each fee component of a trip.}
\label{table-base-fare}
\end{table}
Table~\ref{table-base-fare} shows the cost for each fee component for a single-passenger trip used in our experiment.
The choice of the fees is validated by examining the TNP dataset as follows.
Let $Z(h,x,y)$ be the set of trips with pick-up times during hour $h$, origins (pick-up areas) in region $x$, and destinations (drop-off areas) in region $y$.
For each set $Z(h,x,y)$ of trips, we calculate the average fare $avg(Z(h,x,y))$ paid by the passengers of these trips, namely, sum the fare of the trips in $Z(h,x,y)$ and divide it by $|Z(h,x,y)|$.
Then, we examine the periods of time containing many trips during which surge pricing is unlikely to occur (before noon, as demonstrated by~\cite{Yan-NRL20}).
We compare the estimated average fare $f'(Z(h,x,y))$, calculated using the chosen fee component costs in Table~\ref{table-base-fare}, against the the average fare $avg(Z(h,x,y))$ for $h\in\{10,11\}$.
From the stats in Table~\ref{table-average-fare-ratio}, the cost for each fee component shown in Table~\ref{table-base-fare} is reasonable.
\begin{table}[htbp]
\small
\centering
\begin{tabularx}{0.78\textwidth}{ c | *{2}{Y} }
\toprule
         \multirow{2}{*}{Ratio range}                        & \multicolumn{2}{c}{Percentage of $\cup_{x,y} Z(h,x,y)$ fall in the range}       \\ 
                                                                                 & $h=10$ & $h=11$ \\ 
\midrule
         $\frac{avg(Z(h,x,y))}{f'(Z(h,x,y))} > 1.25$   & 8.64\% & 13.48\%    \Bstrut    \\
         $1.25 \leq \frac{avg(Z(h,x,y))}{f'(Z(h,x,y))} \leq 0.75$   & 84.48\% & 82.02\%  \TBstrut \\
         $0.9 \leq \frac{avg(Z(h,x,y))}{f'(Z(h,x,y))} \leq 1.1$   & 40.32\% & 38.52\%    \TBstrut   \\
         $\frac{avg(Z(h,x,y))}{f'(Z(h,x,y))} < 0.75$   & 6.88\% & 4.49\%    \Tstrut   \\ 
\bottomrule
   \end{tabularx}
\captionsetup{font=small}
\caption{Percentage of $\cup_{x,y} Z(h,x,y)$ that fall in different ranges of $avg(Z(h,x,y))/f'(Z(h,x,y))$.}
\label{table-average-fare-ratio}
\end{table}
Note that fares are rounded to the nearest \$2.50, so some inaccuracy may occur due to this.
The ratio $\gamma(h,x,y) = avg(Z(h,x,y)) / f'(Z(h,x,y))$ is also the estimated average surge pricing factor (same surge pricing factor $\gamma(r_j)$ mentioned above) for a passenger $r_j$ with an origin in region $x$ and a destination in region $y$ that is picked-up during hour $h$.
Promotion discounts are given out regularly by ridesharing companies, which can cause the fare to be lower than normal.
We simulate this type of discount by applying a surge pricing factor $\gamma(h,x,y) < 1$ to any passenger trip $r_j$ in $Z(h,x,y)$, causing $f(\eta_i,r_j)$ to be lower.

We use a similar process to estimate an average vehicle (driving) speed.
For each set $Z(h,x,y)$ of trips, we calculate the average speed $spd(h,x,y)$ by these trips, namely, the sum of the distance of trips in $Z(h,x, y)$ is divided by the sum of the duration of the trips in $Z(h,x,y)$.
In this way, the average speed $spd(h,x,y)$ gives variable vehicle speeds according to the time $h$ of day from region $x$ to region $y$.
Since the TNP does not contain exact coordinates for pick-up and drop-off locations and travel time is required for profit/revenue calculations, we estimate the travel time during hour $h$ from a location $l_a$ in region $x$ to location $l_b$ in region $y$ as $t(l_a,l_b) = \dist(l_a,l_b) / spd(h,x,y)$.
Note that for both estimated surge pricing factor $\gamma(h,x,y)$ and $spd(h,x, y)$, $x=y$ is also considered.

From the TNP dataset, we calculate the average amount of tips and the percentage of trips which a tip is given as follows.
The distance of a trip is rounded to the nearest mile; and we denote the set of trips with distance $d$ as $Z_d$.
Let $Z^+_d$ be the set of trips with tips given and $\epsilon(r_j,d)$ be the average amount of a tip given by a passenger $r_j$ for trip distance $d$.
Figure~\ref{fig-tips} shows $\epsilon(r_j,d)$ and $p(d) = |Z^+_d| / |Z_d|$ for each distance $d$ round to the nearest mile.
\begin{figure}[hbtp]
\centering
\includegraphics[width=0.65\linewidth]{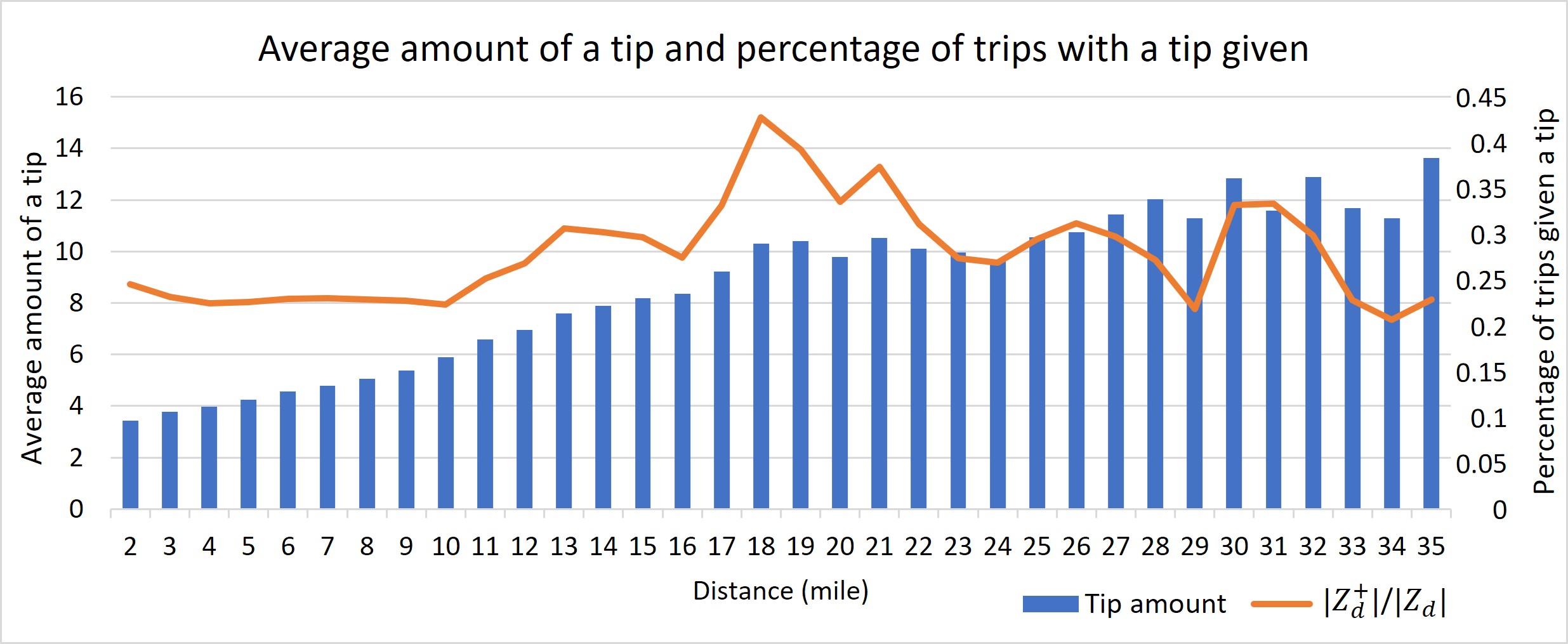}
\captionsetup{font=small}
\caption{Average amount of a tip $\epsilon(r_j,d)$ and ratio $|Z^+_d| / |Z_d|$ for each rounded distance $d$.}
\label{fig-tips}
\end{figure}
We use this result to determine a tip given by a passenger.
Given a match $(\eta_i, R_i)$, $\SFP(\eta_i, R_i)$ and a passenger $r_j \in R_i$, let $d_{r_j}$ be the nearest mile of the subpath $(o_j,d_j)$ in $\SFP(\eta_i, R_i)$.
Then $\epsilon(r_j,d_{r_j})$ is the average amount of tip from $r_j$ (blue bar in Figure~\ref{fig-tips}) if $r_j$ pays a tip, and $p(d_{r_j})$ is the probability $r_j$ pays a tip (orange line in Figure~\ref{fig-tips}).
The average tip amount of a passenger trip $r_j\in R_i$ is $\epsilon(r_j) = p(d_{r_j})\cdot \epsilon(r_j,d_{r_j})$.

The cost $tc(\eta_i, R_i)$ of a match $(\eta_i, R_i)$ for a driver $\eta_i$ is the travel cost for traversing $\SFP(\eta_i, R_i)$.
Then, $tc(\eta_i, R_i) = \dist(\SFP(\eta_i, R_i)) \cdot cost(type)$, where $cost(type)$ is an estimate average cost per mile for a vehicle type under normal traffic conditions.
We use the estimate costs from~\cite{AAA22,NTS22} (which are based on the average gas prices for a 12-month period ending May 2022 in the US): the cost per mile for small Sedan, medium Sedan and medium SUV are \$0.1251, \$0.1437 and \$0.1889 (USD), respectively.

\begin{table}[!htp]
\footnotesize
\centering
   \begin{tabular}{ l | p{13.2cm} }
   	\hline
    $\omega(r_j,R_i)$  & discounted rate for the fee components \\   	
   	$\theta(r_j, R_i)$    & take-rate for a passenger trip $r_j\in R_i$ and is selected uniformly at random from $[\max\{0.05, 0.2\omega(r_j,R_i)\},\max\{0.1,0.25\omega(r_j,R_i)\}]$ \\
    $\gamma(h,x,y)$    & surge pricing factor for a passenger picked-up during hour $h$, from region $x$ to region $y$    \\
    $spd(h,x,y)$    & average vehicle speed during hour $h$ from region $x$ to region $y$    \\
    $\epsilon(r_j)$      & the average amount of a tip given by a passenger $r_j$  \\
 	$tc(\eta_i, R_i)$   & the travel cost for traversing $\SFP(\eta_i, R_i)$ based on a vehicle-type mileage-cost $cost(type)$ \\	
  	$t(l_a,l_b)$        & estimated travel time (duration) of from location $l_a$ to  location $l_b$ in $\SFP(\eta_i, R_i)$       \\ \hline
   \end{tabular}
\captionsetup{font=small}
\caption{Notation used in estimating revenue $rev(\eta_i, R_i)$ and profit $w(\eta_i, R_i)$.}
\label{table-estimation-notation}
\end{table}
Table~\ref{table-estimation-notation} summarizes the notation for all the estimations.
Putting everything together,
the estimated revenue $rev(\eta_i, R_i)$ is
\begin{align*}
\sum_{r_j \in R_i} (1-\theta(r_j, R_i)) &\cdot \omega(r_j,R_i) \cdot \gamma(r_j) \cdot \\
&(1.8 + \sum_{j_1\leq a \leq j_{q-1}} \frac{0.27\cdot t(l_a,l_{a+1}) + 0.8\cdot \dist(l_a,l_{a+1})}{np(l_a,l_{a+1})}) + \epsilon(r_j),
\end{align*}
and the estimated profit $w(\eta_i, R_i)$ after all of $R_i$ are served is
\[
w(\eta_i, R_i) = rev(\eta_i, R_i) - tc(\eta_i, R_i) = rev(\eta_i, R_i) - \dist(\SFP(\eta_i, R_i)) \cdot cost(type).
\]

\subsection{Driver and passenger trips generation} \label{sec-experiment-instance}
For each 15-minute time interval $h_t$, $1\leq t\leq 4$, in hour $h$, we first generate a set of passengers and then a set of drivers.
Passengers are generated according to the average number of trips occurred per hour, calculated using the TNP dataset.
Let $Z(h_t,x,y)$ be the set of trip records in the TNP dataset with pick-up time in interval $h_t$, origins (pick-up areas) in region $x$, and destinations (drop-off areas) in region $y$.
Let $R(h_t,x,y)$ be the set of passengers generated for interval $h_t$ with origins in $x$ and destinations in $y$.
Then, $|R(h_t,x,y)| = \ceil{|Z(h_t,x,y)|/days(h_t,x,y)}$, where $days(h_t,x,y)$ is the number of weekdays in TNP that contain at least a trip record of $Z(h_t,x,y)$.
After $R(h_t,x,y)$ is generated, a set $D(h_t,x)$ of drivers with origins in region $x$ is generated.
The destination region $y$ of a driver $\eta_i \in D(h_t,x)$ is decided as follows.
Let $sum(h_t,x) = \sum_y |R(h_t,x,y)|$ for each region $x$. Then, $\eta_i$ has a destination in region $y$ with probability $|R(h_t,x,y)| / sum(h_t,x)$.
For any driver $\eta_i \in D(h_t,x)$ or passenger $r_i \in R(h_t,x,y)$, the actual origin $o_i$ and destination $d_i$ are two random locations in regions $x$ and $y$, respectively.
The only exceptions are regions R25 and R20 (where the O'Hare International and Midway airports are located, respectively); all drivers' and passengers' origins and destinations in R25 are the O'Hare airport, and there is 50\% chance that origins and destinations in R20 are the Midway airport.
The ridesharing instance in interval $h_t$ consists of $D=\cup_{x} D(h_t,x)$ and $R=\cup_{x,y} R(h_t,x,y)$.

For variant \textbf{RPC1}, we set $0.9\leq \ceil{\frac{|D(h_t,x)|}{|R(h_t,x,y)|}} \leq 1.1$, depending on the hour of the day.
For variant \textbf{RPC\texttt{+}}, $|D(h_t,x)| = \ceil{\frac{|R(h_t,x,y)|}{4}}$ for each $h_t$ during peak hours (7:00-9:59) and (16:00-19:59); and $\ceil{\frac{|R(h_t,x,y)|}{3}} \leq |D(h_t,x)| \leq \ceil{\frac{|R(h_t,x,y)|}{2}}$ during non-peak hours.
Figure~\ref{fig-number-trips} shows the number of drivers and passenger generated for each interval.
\begin{figure}[b]
\centering
\includegraphics[width=0.55\linewidth]{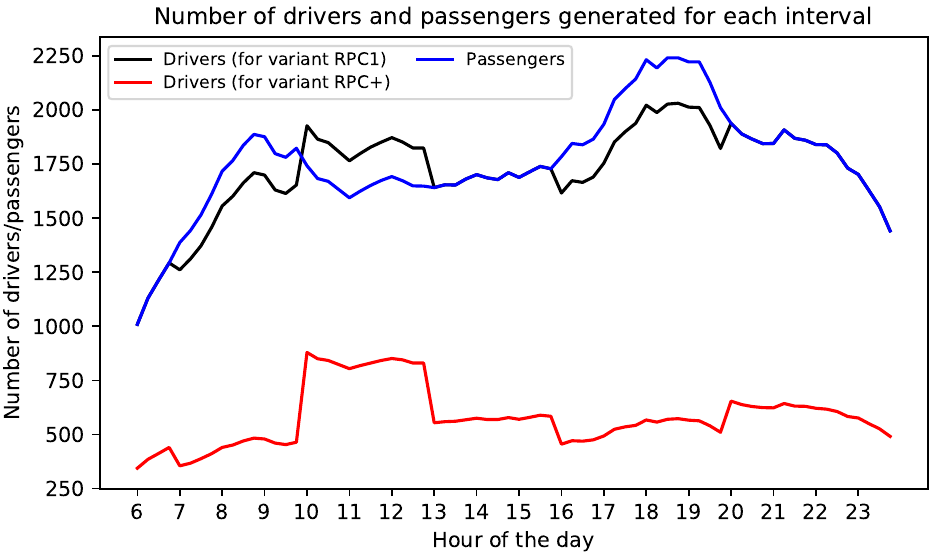}
\captionsetup{font=small}
\caption{The number of drivers and passenger generated for each interval.}
\label{fig-number-trips}
\end{figure}
These numbers are selected with the consideration of surge pricing factor and capacities of drivers' vehicles.
Since the trip records in the TNP dataset are for served (completed) trips, the number of drivers should not be too low.
We consider only Sedan vehicle types with capacity for variant RPC1.
For variant RPC\texttt{+}, we consider three vehicle types: small and medium Sedans with capacity [1,3] and SUVs with capacity [1,5] inclusive.
During peak hours, roughly 95\% and 5\% of vehicles have capacities randomly selected from Sedan and SUV, respectively.
During non-peak hours, roughly 90\% and 10\% of vehicles have capacities randomly selected from Sedan and SUV, respectively.
Other parameters for drivers and passengers are listed in Table~\ref{table-parameters}.
\begin{table}[!ht]
\footnotesize
\setlength\tabcolsep{3pt}
\centering
   \begin{tabular}{ l | p{9.85cm} }
   	\hline
    $\SP(o,d)$                                  & shortest path from location $o$ to location $d$    \\
    Earliest departure time $\alpha_i$  & immediate to end of a time interval $h_t$ for any driver $\eta_i$ or passenger $r_i$ in $D(h_t,x) \cup R(h_t,x,y)$ \\
    Driver detour limit $z_i$ for $\eta_i\in D(h_t,x)$     & at most $\max\{\frac{a\cdot \dist(\SP(o_i,d_i))}{spd(h,x,y)}, 45\}$, $a\in[1.2,1.4]$ 	\\
  	Latest arrival time for $\eta_i\in D(h_t,x)$     & $\alpha_i + a(\frac{\dist(\SP(o_i,d_i))}{spd(h,x,y)} + z_i)$, $a\in[1.0,1.25]$      \\
    Max travel duration for $\eta_i\in D(h_t,x)$   & $\frac{\dist(\SP(o_i,d_i))}{spd(h,x,y)} + z_i$   		\\
    Latest arrival time for $r_j\in R(h_t,x,y)$         & $\alpha_j + \frac{a\cdot \dist(\SP(o_j,d_j))}{spd(h,x,y)}$, $a\in[2.0,3.0]$ \\
    Max travel duration for $r_j\in R(h_t,x,y)$       & $\frac{a\cdot \dist(\SP(o_j,d_j))}{spd(h,x,y)}$, $a\in[1.5,2.0]$  	\\ 
    \hline
   \end{tabular}
\captionsetup{font=small}
\caption{Parameters for drivers and passengers.}
\label{table-parameters}
\end{table}

Feasible matches are computed from $D$ and $R$ in each interval $h_t$.
For a driver $\eta_i$, a feasible match $(\eta_i,R_i)$ with $|R_i|=1$ is called a \emph{base match} because for any feasible match $(\eta_i, R'_i)$ with $|R'_i|>1$, the match $(\eta_i, R_i)$ with $R_i \subset R'_i$ must exist~\cite{Gu-ISAAC21,Stiglic-COR18}.
We compute base matches first and then feasible matches with $|R_i|>1$ as in~\cite{Alonso-Mora-PNAS17,Simonetto-TRCET19}.
Shortest paths in our simulation are computed in \textit{real-time}.
To speedup the computation of feasible matches for practical reasons, we first apply a conditional check to see if a driver $\eta_i$ and a passenger $r_j$ is a \emph{candidate pair}.
That is, test $\eta_i$ and $r_j$ should be considered in a base match by estimating the travel distance without computing any shortest path.
We also limit the number of base matches and total feasible matches a driver $\eta_i$ can have.
\begin{enumerate}\setlength\itemsep{0em}
\item Let $\overline{\dist}(o,d)$ be the straight-line distance from location $o$ to location $d$. Let $ed(\eta_i,r_j)=\tau\cdot (\overline{\dist}(o_i,o_j)+\overline{\dist}(o_j,d_j)+\overline{\dist}(d_j,d_i))$ be the estimated travel distance, for some $\tau > 0$.
If the maximum travel distance for $\eta_i$ is at least $ed(\eta_i,r_j)$, then the match $(\eta_i, R_i=\{r_j\})$ is a candidate.
Otherwise, $(\eta_i, \{r_j\})$ is unlikely a feasible match (so $(\eta_i,\{r_j\})$ is not checked).
\item Any driver $\eta_i$ can have at most 100 base matches and at most 500 feasible matches in total; and each passenger can belong to at most 20 base matches (so 20 different drivers).
\end{enumerate}

\subsection{Computational results} \label{sec-experiment-results}
All algorithms were implemented in Java, and the experiments were conducted on an Intel Core i7-6700 processor with 2133 MHz of 12 GBs RAM available to JVM.
All ILP formulations in our algorithms are solved by CPLEX v12.10.1.
We label the algorithm CPLEX uses to solve ILP formulations~\eqref{formulation-max-passenger}-\eqref{constraint-profit} and \eqref{formulation-max-passenger}-\eqref{constraint-nonnegative-edge} for RPC1 and RPC+ by \textbf{Exact}.
Note that by default, CPLEX uses multithreading, and we leave it as it is.
Recall that for RPC1, the implementation use ILPs is labeled as \textbf{ExactNF1}, the implementation based on graph algorithms is labeled as \textbf{ExactNF2} and the greedy $\frac{1}{2}$-approximation algorithm is labeled as \textbf{Greedy}.
A passenger $r_j \in R$ is called \emph{served} if $r_j \in R_i$ such that $(\eta_i, R_i)$ is a feasible match belongs to a solution computed by one of the algorithms.

\subsubsection{RPC1 results}
The base case instances use the profit calculation described in Section~\ref{sec-experiment-profit}.
The estimated distance factor used is $\tau=0.6$ for the computation heuristic described in~\ref{sec-experiment-instance}.
The overall results are shown in Table~\ref{table-base-result} for profit targets $c_1 = w(M')$, $c_2 = 0.8\cdot w(M')$ and $c_3 = 0.6\cdot w(M')$, where $M'$ is a maximum weight matching in $H$ for each interval.
\begin{table}[!ht]
\footnotesize
\setlength\tabcolsep{4pt}
\centering
   \begin{tabular}{ l | c | c | c | c | c}
   	\hline
    \multicolumn{2}{ l |}{}    & Greedy     & ExactNF1    & ExactNF2	    & Exact   \\ \hline
   \multirow{3}{4.7cm}{Total \# of passengers served in all intervals}  & ($c'_1$=\$1587436)   & 109770     & 109771    & 109771     & 109771  \\    
    & ($c'_2$=\$1269949)       & 109775 	  & 110035     & 110035    & 110035     \\    
    &  ($c'_3$=\$952462)       & 109775     & 110035     & 110035    & 110035      \\ \hline    
  \multirow{3}{4.7cm}{Total profit of served matches in all intervals}  & ($c'_1$=\$1587436)   & 1587436    & 1587436   & 1587436    & 1587436  \\
   & ($c'_2$=\$1269949)        & 1587432 	  & 1586707     & 1586707    & 1465676    \\ 
   &  ($c'_3$=\$952462)        & 1587432    & 1586707     & 1586707    & 1457338    \\ \hline     
    \multirow{3}{4.7cm}{Avg running time (second) per interval}   & ($c'_1$=\$1587436)  & 5.573	   & 6.249  & 4.765     & 7.030 \\
    & ($c'_2$=\$1269949)       & 5.670       & 6.223      & 4.484		  & 6.591    \\
    &  ($c'_3$=\$952462)       & 5.765       & 6.379      & 4.565      & 6.298 	  \\ \hline
    \multicolumn{3}{ l |}{Avg running time to compute the matches per interval} & \multicolumn{3}{l}{238.71 seconds} \\
    \multicolumn{3}{ l |}{Total number of drivers and passengers generated, respectively} & \multicolumn{3}{l}{124340 and 126625}\\ \hline
   \end{tabular}
\captionsetup{font=small}
\caption{Performances of algorithms for RPC1 on base case instances. For $1\leq a\leq 3$, $c'_a=\sum^{18}_{h=1}\sum^{4}_{h_t=1} c_a$ (in dollar).}
\label{table-base-result}
\end{table}
Due to the profit calculation, there are only 7.28 negative-profit matches per interval on average, which is about 0.019\% of the average number of matches per interval (37939.99).
This is by design for drivers to make money in practice, which results in the excellent performance of Greedy, as optimal solutions have very few negative-profit matches.
The Greedy solutions serve about 99.76\% of passengers served by the optimal solutions.
In some intervals, Greedy solutions have higher profits due to negative matches are assigned in the optimal solutions.
From Table~\ref{table-base-result}, ExactNF2 runs faster than other algorithms, and ExactNF1/2 always produce optimal solutions with the highest profits of all optimal solutions.

The base case instances may not truly reflect the whole picture when retail gas price increases and traffic congestion occurs.
The gas price (regular) in Chicago was increased by nearly 35\% from March 2022 to June 2022 and nearly 80\% from Dec. 2021 to June 2022\footnote{Energy Information Administration. \url{https://www.eia.gov}}.
According to~\cite{Litman-TCBA16}, fuel consumption can increase 30\% under heavily congestion.
Note that the cost increase from these two together are multiplicative.
Although the cost due to congestion is compensated by a higher revenue (longer travel time), some costs caused by congestion are not recovered.
We considered five different settings for travel cost increases due to extra fuel cost (gas price + congestion):
\begin{itemize}\setlength\itemsep{0em}
\begin{small}
\item 0-20\% increase (0\% for non-peak hours and 20\% for peak hours), 20-40\%, 40-60\%, 60-80\% and 80-100\%.
\end{small}%
\end{itemize}%
In addition to gas price fluctuation, other major operating costs for drivers include maintenance, depreciation, extra insurance and tax, especially for the drivers that provide frequent ridesharing service.
From the findings in~\cite{AAA22,NTS22}, we add the following maintenance + depreciation (based on 20k miles/year) operating costs, labeled as \textit{OP}, to $tc(\eta_i,R_i)$ for each match $(\eta_i, R_i)$:
\begin{itemize}\setlength\itemsep{0em}
\begin{small}
\item \$0.0887 + \$0.1851 per mile for Small Sedan and \$0.1064 + \$0.2505 per mile for Medium Sedan.
\end{small}%
\end{itemize}%
For all other costs, an extra 20\%/40\% cost increase is applied to $tc(\eta_i, R_i)$ for each match $(\eta_i, R_i)$.
Altogether, we tested the following six cost settings (fuel and operating costs) added to the base case:
\begin{itemize}\setlength\itemsep{0em}
\begin{small}
\item \textit{S1} (20-40\% cost increase + operating costs OP), \textit{S2} (40-60\% + OP),
\textit{S3} (60-80\% + OP), \textit{S4} (80-100\% + OP), \textit{S5} (100-120\% + OP), \textit{S6} (120-140\% + OP).
\end{small}
\end{itemize}%
The results of these six settings are depicted in Table~\ref{table-cost-increased-result} for profit target $c = 0.8\cdot w(M')$ for each interval.
\begin{table}[!ht]
\footnotesize
\setlength\tabcolsep{4pt}
\centering
\begin{tabular}{l | l | c | c | c | c}
\hline
	\multicolumn{2}{ l |}{} 		& Greedy    & ExactNF1   & ExactNF2   & Exact	    \\ \hline
	\multicolumn{2}{c |}{\#S1}  	& 107233    & 110035     & 110035     & 110035      \\
    \multicolumn{2}{c |}{\#S2}    & 106953   	 & 110035     & 110035     & 110035      \\
    \multicolumn{2}{c |}{\#S3}    & 106658   	 & 110035     & 110035     & 110035     \\
    \multicolumn{2}{c |}{\#S4}    & 106387     & 110035     & 110035     & 110035 	 \\
    \multicolumn{2}{c |}{\#S5}    & 106081   	 & 110035     & 110035     & 110035     \\   
    \multicolumn{2}{c |}{\#S6}    & 105795     & 110035     & 110035     & 110035     \\
\midrule
	\$S1  & $c'$=\$931569 	      & \$1162693	 & \$1150650  & \$1150650  & \$1006764     \\
    \$S2  & $c'$=\$907483        & \$1132014	 & \$1118309  & \$1118309  & \$977763      \\
    \$S3  & $c'$=\$883652        & \$1101575	 & \$1086034  & \$1086034  & \$949875      \\
    \$S4  & $c'$=\$860079        & \$1071137	 & \$1053821  & \$1053821  & \$922326		\\
    \$S5  & $c'$=\$836773        & \$1041053   & \$1021658  & \$1021658  & \$897687      \\   
    \$S6  & $c'$=\$813739        & \$1010972   & \$989553   & \$989553   & \$871236      \\
\midrule
	\multicolumn{2}{c |}{$\Theta$S1 (second)}  & 3.316   & 3.976    & 4.066    & 4.577   \\
    \multicolumn{2}{c |}{$\Theta$S2 (second)}  & 3.246   & 3.921    & 4.658    & 4.564   \\
    \multicolumn{2}{c |}{$\Theta$S3 (second)}  & 3.308   & 3.909    & 3.986    & 4.672   \\
    \multicolumn{2}{c |}{$\Theta$S4 (second)}  & 3.301   & 3.962    & 4.526    & 4.696   \\
    \multicolumn{2}{c |}{$\Theta$S5 (second)}  & 3.351   & 3.919    & 3.904    & 4.496   \\   
    \multicolumn{2}{c |}{$\Theta$S6 (second)}  & 3.343   & 3.989    & 4.600    & 4.720   \\    
\hline
    \end{tabular}
\captionsetup{font=small}
\caption{(\#) Total number of passengers served and (\$) total profit of served matches in all intervals, and ($\Theta$) average running time per interval for RPC1 using different cost settings. $c'=\sum^{18}_{h=1}\sum^{4}_{h_t=1} c$ (in dollar).}
\label{table-cost-increased-result}
\end{table}
From Tables~\ref{table-base-result}~and~\ref{table-cost-increased-result}, the performance of Greedy is from about 99.76\% (base case) to 96.1\% (the worst case S6) of the exact algorithms in the total number of passengers served.
In the tested instances, Greedy has the fastest average running time in this scenario.
Exact is still the slowest, and ExactNF1 is slightly faster than ExactNF2 on average.
Table~\ref{table-negative-matches} shows results related to negative profits for selected cost settings: S2, S4, and S6.
Solutions produced by Greedy have higher total profits than that of ExactNF1/ExactNF2 and Exact because Greedy solutions have lower number of matches with negative profit, compared to the ExactNF1/ExactNF2/Exact solutions, as stated in Table~\ref{table-negative-matches}.
\begin{table}[!t]
\footnotesize
\setlength\tabcolsep{4pt}
\centering
   \begin{tabular}{ l | c | c | c | c | c}
   	\hline
    \multicolumn{2}{ l |}{}  & Greedy       & ExactNF1   & ExactNF2   & Exact	    \\ \hline
   \multirow{3}{8.25cm}{Avg \% of negative-profit matches served per interval out of all served matches}       
   	& (S2)					  & 0.8219\%     & 1.2332\%   & 1.2314\%	 & 3.5371\%     \\    
    & (S4)                   & 1.2361\%   	  & 1.8867\%   & 1.8867\%	 & 4.6540\%      \\    
    & (S6)                   & 1.8073\%     & 2.8391\%   & 2.8373\%   & 5.9127\%      \\ \hline  
    \multirow{3}{*}{Avg number of negative-profit matches per interval} 
    & (S2) & \multicolumn{4}{c}{3651.903 (9.625\% of total matches)}  \\
    & (S4) & \multicolumn{4}{c}{5013.125 (13.213\% of total matches)}  \\
    & (S6) & \multicolumn{4}{c}{6517.472  (17.178\% of total matches)}  \\ \hline
   \end{tabular}
\captionsetup{font=small}
\caption{Results relate to profit for S2, S4, and S6.}
\label{table-negative-matches}
\end{table}
Of course, ExactNF always produces optimal solutions.
As a result, ExactNF1 and ExactNF2 always outperform Exact.
As the number of matches with negative profit increases (S1$\rightarrow$S6), Greedy serves less passengers, and the exact algorithms stay the same.
These may suggest that when there are more matches with negative profit, ExactNF1 is a better choice as it runs slightly faster than ExactNF2 on average.
On the other hand, if profit is also important, using Greedy is acceptable since it produces the highest total profits and its performance is about 96.1\% of the exact algorithms in the total number of passengers served.

The mean occupancy rate (in each interval) for exact algorithms stays the same for all six settings S1 to S6.
The mean occupancy rate is calculated as, in each interval, (the number of served passengers +
the total number of drivers) divided by the total number of drivers.
The occupancy rates for each time interval in S2 and S6 are shown in Figure~\ref{fig-occupanyRateRPC1}.
\begin{figure}[!th]
\centering
\includegraphics[width=0.6\linewidth]{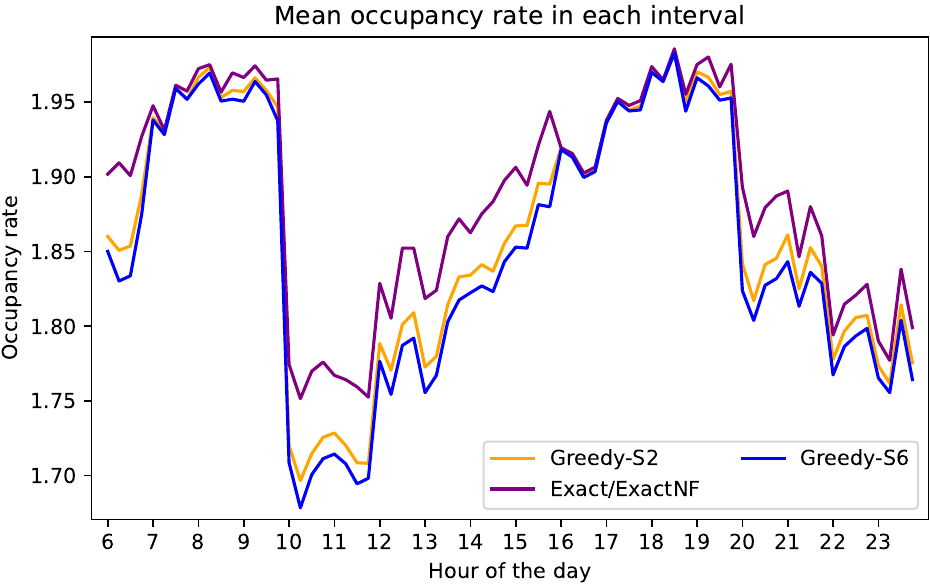}
\captionsetup{font=small}
\caption{The mean occupancy rate in each interval for S2 and S6.}
\label{fig-occupanyRateRPC1}
\end{figure}
The mean occupancy rates follow a similar distribution as the number of passengers generated.
The average occupancy rate for exact algorithms (calculated as the sum of the mean occupancy rate for each interval and divided it by 72) is 1.8868.
When there are many drivers and passengers generated (balanced supply and demand), many of them are matched together.
During morning and afternoon peak/rush hours, the mean occupancy rates are close to 2, meaning most of the generated drivers are assigned a passenger.
It may be beneficial to decrease the the parameter $\tau$ for candidate pair checks when supply/demand are lower in some time intervals, so more driver-passenger pairs are considered.
Finally, from Table~\ref{table-base-result}, the total maximum number of passengers served is at 110035 for $c_2$ and $c_3$ (same in Table~\ref{table-cost-increased-result} for $c = 0.8\cdot w(M')$).
This may indicate that $c_2 = 0.8\cdot w(M')$ is a relative stable choice for profit target in general.

\subsubsection{RPC\texttt{+} results}
The base case instances use the profit calculation described in Section~\ref{sec-experiment-profit}.
The estimated distance factor used is $\tau=0.8$ for the computation heuristic described in~\ref{sec-experiment-instance}.
We assume all passengers are willing to participant in ridesharing.
Recall that the $\frac{2}{3\lambda}$-approximation algorithm that solves RPC\texttt{+} is labeled as \textbf{LS2} and the first step of LS2 is labeled as \textbf{SimpleGreedy}.

For this variant, the profit target $c$ is upper bounded by the weight $w(M')$ of the matching $M'$ found by SimpleGreedy.
Recall that $A =\{e \in M' \mid |R(e)|=1\}$, as defined in the description of Algorithm LS2, and
Assumption~\ref{assumption-profit-target-range} requires $c \leq w(M' \setminus A) + 2w(A) / (\lambda+1)$.
We set a lower bound $LB=\min\{w(M' \setminus A) + 2w(A) / (\lambda+1), 0.6w(M')\}$.
We tested three profit targets $c_1=w(M')$, $c_2 = 0.5\cdot(w(M') - LB) + LB$ and $c_3=LB$.
The overall results are shown in Table~\ref{table-base-result-rpc+}.
\begin{table}[ht]
\footnotesize
\setlength\tabcolsep{3pt}
\centering
   \begin{tabular}{ l | c | c | c | c}
   	\hline
    \multicolumn{2}{ l |}{}   & SimpleGreedy    & LS2        & Exact	    \\ \hline
\multirow{3}{*}{Total \# of passengers served in all intervals}   & ($c'_1$=\$845817)
								& 63554 	      & 64099      & 71197       \\    
    & ($c'_2$=\$676653)        & 63554  	      & 64118      & 71208          \\    
    & ($c'_3$=\$507490)        & 63554          & 64118      & 71208          \\ \hline  
   \multirow{3}{*}{Total profit of served matches in all intervals} & ($c'_1$=\$845817) 
   								& 845817          & 848677     & 846893   \\
   & ($c'_2$=\$676653)         & 845817  	      & 848271      & 702130          \\ 
   & ($c'_3$=\$507490)         & 845817          & 848271     & 681472          \\ \hline  
    \multirow{3}{*}{Avg running time (second) per interval}   & ($c'_1$=\$845817)  
    							 & 0.0445	       & 0.0761     & 47.880   \\
    & ($c'_2$=\$676653)         & 0.0386         & 0.0708     & 33.279   \\
    & ($c'_3$=\$507490)         & 0.0397         & 0.0695     & 35.664   \\ \hline
    \multicolumn{2}{ l |}{Avg number of feasible matches per interval} & \multicolumn{3}{l}{103612.431} \\
    \multicolumn{2}{ l |}{Avg number of matches with negative profit per interval} & \multicolumn{3}{l}{4.236 (0.0041\% of total matches per int.)} \\
    \multicolumn{2}{ l |}{Avg running time to compute the matches per interval} & \multicolumn{3}{l}{389.176 seconds} \\
    \multicolumn{2}{ l |}{Total number of drivers and passengers generated respectively} & \multicolumn{3}{l}{40573 and 126625}\\ \hline
   \end{tabular}
\captionsetup{font=small}
\caption{Performances of algorithms for RPC\texttt{+} on base case instances. For $1\leq a\leq 3$, $c'_a=\sum^{18}_{h=1}\sum^{4}_{h_t=1} c_a$ (in dollar).}
\label{table-base-result-rpc+}
\end{table}
The performances of SimpleGreedy and LS2 are about 89.25\% and 90.04\% of the exact algorithm (Exact), in the total number of passengers served.
The running time of LS2 is only 0.0312 second longer than that of SimpleGreedy on average.
On the other hand, the running time of Exact is 470-630 times longer than that of LS2, depending on the profit target.
For very larger instance thought, Exact may not be suitable for real-time computation.
Nonetheless, the overall running times of all the algorithms for the tested cases are practical enough as shown in Table~\ref{table-base-result-rpc+}.

In terms of occupancy rates, Exact has the best mean occupancy rate in each interval as expected, and its mean occupancy rate is also more stable compared to the other two algorithms (see an example for $c_1$ in Figure~\ref{fig-occupanyRateRPC+}).
\begin{figure}[!b]
\centering
\includegraphics[width=0.61\linewidth]{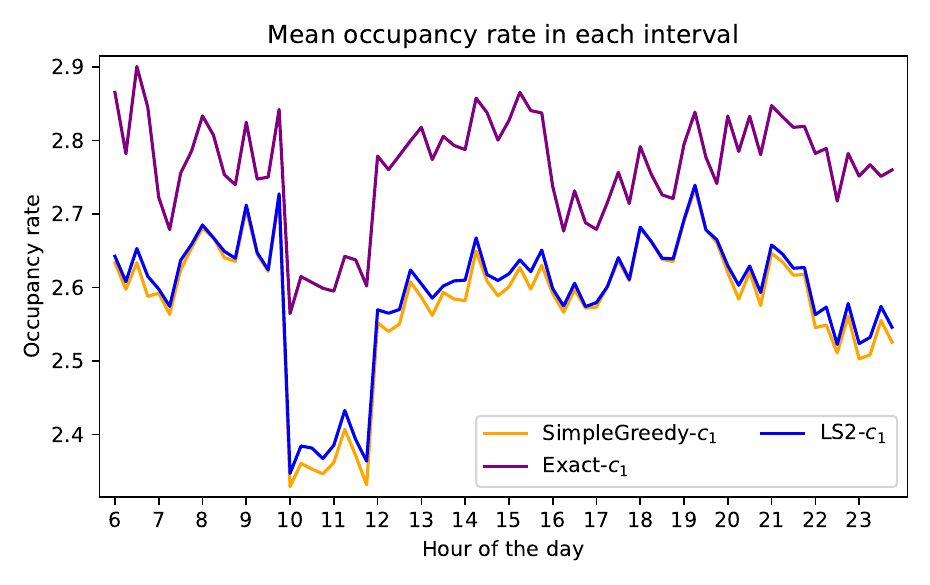}
\captionsetup{font=small}
\caption{The mean occupancy rate in each interval for RPC\texttt{+} and $c_1$.}
\label{fig-occupanyRateRPC+}
\end{figure}
In most intervals, the mean occupancy rate of Exact is 2.7 or higher, and the mean occupancy rates of LS2 and SimpleGreedy are close to 2.6.
The average occupancy rate for each algorithm is depicted in Table~\ref{table-RPC+-OR}.
This aligns with previous studies that there is potential in ridesharing.
Our experiment further shows that there is potential in profit-maximizing MoD platforms by utilizing ridesharing as profit targets are achieved by our algorithms.
It is beneficial to use different algorithms for different intervals (depending on the supply and demand) if time limit for computing a solution is important.
\begin{table}[t]
\small
\captionsetup{font=small}
\caption{The average occupancy rate.}
\centering
   \begin{tabular}{ l | c | c | c | c }
   	\hline
    \multicolumn{2}{c|}{}               			    & SimpleGreedy  & LS2     & Exact   \\ \hline
    \multirow{3}{*}{Average occupancy rate}  & $c_1$  & 2.5808        & 2.5933  & 2.7628  \\
    &											 $c_2$  & 2.5808        & 2.5938  & 2.7631  \\
    &											 $c_3$  & 2.5808        & 2.5938  & 2.7631  \\ \hline    
   \end{tabular}
\label{table-RPC+-OR}
\end{table}

\subsubsection{Discussion}
Based on the RPC1 and RPC\texttt{+} results, our algorithms are effective for achieving the optimization goal in practical scenarios.
Although the RPC problem is NP-hard, the exact algorithms are efficient enough to find optimal solutions, which can be practical (depending on the applications).
The approximation algorithms Greedy and LS2 can achieve 96.1\% and 90.04\% of the optimal solutions to RPC1 and RPC\texttt{+}, respectively, in the number of passengers served.
The average occupancy rates of the solutions to RPC1 and RPC\texttt{+} are higher than the reported occupancy rate in the US (which was 1.5 in 2017) for all algorithms tested.
Our experiment results suggest that there is potential in profit-maximizing/profit-incentive MoD platforms by utilizing ridesharing.

\begin{table}[!ht]
\small
\centering
   \begin{tabular}{ c | c | c | c | c | c }
   	\hline
    	& RPC1 base   & RPC1: S2    & RPC1: S4    & RPC1: S6    & RPC\texttt{+} base   \\ \hline   
\multirow{2}{*}{*}
		& \#109770    & \#106074    & \#105072    & \#103883    & \#65913     \\
     	& \$1587436   & \$1134354   & \$1075099   & \$1017174   & \$893879  \\ \midrule
\multirow{2}{*}{ $\diamond$}
		& \#109771    & \#110035    & \#110035    & \#110035    & \#71197   \\
     	& \$1586707   & \$1118309   & \$1053821   & \$989553    & \$846893  \\ \hline
   \end{tabular}
\captionsetup{font=small}
\caption{Total number \# of passengers served and total profit \$ of served matches in all intervals. (*) Optimal solutions to RP. ($\diamond$) Optimal solutions to RPC1 (for ExactNF and $c_2$) and RPC\texttt{+} (for Exact and $c_1$).}
\label{table-RP-result}%
\end{table}%
We also computed optimal solutions to the RP problem (formulation~(\ref{formulation-sharingmodel})-\eqref{constraint-disjoint-RP}) using some RPC1 and RPC\texttt{+} test instances. The results are shown in Table~\ref{table-RP-result}. 
From Table~\ref{table-RP-result}, optimal solutions to RPC1 and RPC\texttt{+} serve 5.92\% (S6) and 8.02\% (RPC\texttt{+}) more passengers than the respective RP optimal solutions.
The profits of the optimal solutions for RPC1 and RPC\texttt{+} are reasonably close to that of the RP optimal solutions when there are not many matches with negative profit.
RPC can be an alternative to RP in practice since more passengers are served with a controllable profit target, which can be adjusted for each interval.
MoDs can choose to compensate the drivers that serve negative-profit matches (e.g., lower the take-rate).

\section{Conclusion} \label{sec-conclusion}
Our model of the RPC problem provides a new framework to incorporate a flexible pricing scheme to maximize the number of passengers served while meeting a profit target.
Our approach for solving the RPC problem allows personal/ad-hoc drivers and designated drivers to participate in the ridesharing system at the same time.
The RPC problem is a more complex variant of the maximum set packing problem, and hence, it is NP-hard.
We give a polynomial time exact algorithm (with two different practical implements) and approximation algorithms for special cases (RPC1 and RPC\texttt{+}) of the problem.
Experimental results show that practical profit (price) schemes can be incorporated into our model and suggest that there is potential in profit-maximizing/profit-incentive MoD platforms by utilizing ridesharing.

The RPC\texttt{+} variant considers only matches with non-negative profit, which may cover the MoD systems' profit-incentive, but it may impose a limit on improving the number of passengers served.
It is worth developing algorithms for more general cases where matches with negative profit are also considered.
A related optimization problem is to maximize the system-wide profit while a number of passengers must be served.
Such an optimization problem may satisfy more demand compared to the RPC problem, which is important in reducing congestion and $\textup{CO}_2$ emissions.

\bibliographystyle{plainurl}
{\small
\bibliography{rpc}
}

\end{document}